\documentclass[11pt]{article}

\usepackage{fullpage}

\usepackage{graphicx}
\graphicspath{{../}}

\usepackage{amsmath}
\usepackage{amsthm}
\usepackage{amssymb}
\usepackage{natbib}
\usepackage{mathtools}
\usepackage[hidelinks]{hyperref}
\usepackage{bm}

\usepackage{amsfonts,amssymb,latexsym}
\usepackage{graphicx}
\usepackage{color}

\usepackage{enumitem}

\usepackage[noabbrev]{cleveref}

\usepackage{amsthm}
\usepackage{thmtools,thm-restate}

\usepackage[dvipsnames]{xcolor}

\usepackage{pgfplots}
\usetikzlibrary{patterns}
\usepgfplotslibrary{fillbetween}
\usepackage{xfrac}
\usepackage{framed}
\usepackage[stable]{footmisc}

\pgfplotsset{compat=1.17}

\newtheorem{theorem}{Theorem}[section]
\newtheorem{definition}[theorem]{Definition}
\newtheorem{lemma}[theorem]{Lemma}
\newtheorem{proposition}[theorem]{Proposition}
\newtheorem{corollary}[theorem]{Corollary}

\newtheorem{example}[theorem]{Example}

\DeclareMathOperator*{\argmax}{argmax}

\newtheorem{fact}[theorem]{Fact}

\newcommand{\agind}[1][i]{_{#1}}

\newcommand{\noaccents}[1]{#1}

\newcommand{\newagentvar}[3][\noaccents]{%
\expandafter\newcommand\expandafter{\csname #2\endcsname}{#1{#3}}%
\expandafter\newcommand\expandafter{\csname #2s\endcsname}{#1{\boldsymbol{#3}}}%
\expandafter\newcommand\expandafter{\csname #2smi\endcsname}[1][i]{#1{\boldsymbol{#3}}_{-##1}}%
\expandafter\newcommand\expandafter{\csname #2i\endcsname}[1][i]{#1{#3}\agind[##1]}%
\expandafter\newcommand\expandafter{\csname #2ith\endcsname}[1][i]{#1{#3}_{(##1)}}%
}

\newcommand{\newitemvar}[3][\noaccents]{%
\expandafter\newcommand\expandafter{\csname #2\endcsname}{#1{#3}}%
\expandafter\newcommand\expandafter{\csname #2s\endcsname}{#1{\boldsymbol{#3}}}%
\expandafter\newcommand\expandafter{\csname #2smj\endcsname}[1][j]{#1{\boldsymbol{#3}}_{-##1}}%
\expandafter\newcommand\expandafter{\csname #2j\endcsname}[1][j]{#1{#3}_{##1}}%
\expandafter\newcommand\expandafter{\csname #2jth\endcsname}[1][j]{#1{#3}_{(##1)}}%
}

\newcommand{\forrezs}[1]{{#1}^{\rezs}}
\newagentvar[\forrezs]{imbal}{\phi}
\newagentvar[\forrezs]{cumimbal}{\Phi}

\newagentvar{alloc}{x}
\newagentvar{Dalloc}{\Delta x}

\newagentvar{price}{p}
\newagentvar{balloc}{y}

\newagentvar{util}{u}

\newagentvar{dist}{F}
\newagentvar{dens}{f}
\newagentvar{hazard}{h}
\newagentvar{cumhazard}{H}

\newcommand{\game}{\mathcal{G}}

\newcommand{\reals}{{\mathbb R}}

\newagentvar{skill}{s}
\newcommand{\subsy}{d}
\newcommand{\gamedy}{\mathcal{D}}
\newcommand{\gamecor}{\mathcal{H}}
\newcommand{\subsylo}{D_-}
\newcommand{\subsyhi}{D_+}
\newcommand{\allocdy}{w}
\newcommand{\quamdy}{\qua^{\ddagger}}

\newcommand{\princ}{P}
\newagentvar{agent}{A}
\newagentvar{budg}{b}
\newcommand{\thresh}{\tau}
\newcommand{\bdist}{B}
\newcommand{\sdist}{S}
\newcommand{\bsupp}{\mathbf{B}}
\newcommand{\ssupp}{\mathbf{S}}
\newcommand{\bcorset}{\mathcal{B}^{\sdist}}
\newagentvar{eff}{e}
\newagentvar{qua}{q}
\newagentvar{quamax}{\qua^{\dagger}}
\newagentvar{quath}{\qua'}
\newagentvar{quaopt}{\qua^*}
\newcommand{\effth}{\eff'}
\renewcommand{\util}{u}
\newcommand{\utilv}{v}
\newcommand{\expecta}{\mathbf{E}}
\newcommand{\eps}{\epsilon}
\newagentvar{att}{a}
\newagentvar{tier}{T}
\newagentvar{dprob}{\pi}
\newcommand{\lo}{L}
\newcommand{\hi}{H}
\newcommand{\figsz}{\scriptsize}
\newcommand{\figszparm}{\scriptsize}

\begin{document}

\begin{titlepage}

\title{Screening with Disadvantaged Agents}

\newcommand{\email}[1]{\href{mailto:#1}{#1}}

\author{Hedyeh Beyhaghi\thanks{Carnegie Mellon University. Email: \email{hedyeh@cmu.edu}} \and Modibo K. Camara \thanks{University of Chicago. 
 Email: \email{mcamara@uchicago.edu}} \and Jason Hartline\thanks{Northwestern University.  Supported in part by NSF CCF 1934931.  Email: \email{hartline@northwestern.edu}} \and Aleck Johnsen\thanks{Geminus Research.  Email: \email{aleck@geminusresearch.com}} \and Sheng Long\thanks{Northwestern University.  Supported in part by NSF CCF 1934931. 
 Email: \email{shenglong2025@u.northwestern.edu}}}

\date{}

\maketitle

Motivated by school admissions, this paper studies screening in a population with both advantaged and disadvantaged agents.
A school is interested in admitting the most skilled students, but relies on imperfect test scores that reflect both skill and effort.
Students are limited by a budget on effort, with disadvantaged students having tighter budgets.
This raises a challenge for the principal: among agents with similar test scores, it is difficult to distinguish between students with high skills and students with large budgets.

Our main result is an optimal stochastic mechanism that maximizes the gains achieved from admitting ``high-skill" students minus the costs incurred from admitting ``low-skill" students when considering two skill types and $n$ budget types. Our mechanism makes it possible to give higher probability of admission to a high-skill student than to a low-skill, even when the low-skill student can potentially get higher test-score due to a higher budget.  Further, we extend our admission problem to a setting in which students uniformly receive an exogenous subsidy to increase their budget for effort.  This extension can only help the school's admission objective and we show that the optimal mechanism with exogenous subsidies has the same characterization as optimal mechanisms for the original problem.
\end{titlepage}



\section{Introduction}
\label{s:intro}


Screening is a problem in which a {\em principal} desires to select only a qualified sub-population of {\em agents} who exceed an appropriate threshold applied to the agents' {\em private types}.

Many real-world problems may be interpreted as special cases of screening problems, including some well-studied problems in standard frameworks (for example within auction design, how to give away an item to an agent who values it the most).  As further examples: school admissions, hiring employees, selecting romantic partners, identifying winners of prestigious awards, qualifying applicants for government-issued licenses, assigning school grades at any level of evaluation (from homework grades to testing grades to overall-course grades), drug-testing, tournament-qualifying, \ldots, all of these and many more scenarios may be modeled as problems of screening.

The challenge of screening is that the principal has only indirect access to the agents' private types, and critically, the agents are either unwilling to reveal their types or are incentivized to take actions that make it difficult for the principal to infer their types.
Since these agents are \emph{strategic}, their private information is only fully or partially elicited by offering appropriate incentives.

This paper considers a screening model of {\em school admissions} where some students may be disadvantaged relative to others.
The school seeks the most skilled students but only has access to an imperfect measure of skill, via test scores.
Relative to their inherent skill, disadvantaged students may perform worse on tests because they have less time to prepare (e.g., due to work obligations or childcare).
Advantaged students may perform better on tests relative to their skill because they have access to additional resources (e.g., a private tutor or test prep).
We model this heterogeneity by assuming that applicants are distinguished both in their skill as well as their \emph{budget} (i.e., how much time and resources they are able to put towards the test).
Students with high skill and high budget are able to excel, provided that they are willing to put in the effort.
However, students with similarly high skill may test poorly if their budgets are too low.

More precisely, we study a mechanism design problem for screening of budgeted agents.  A principal is interested in admitting only an agent with high skill-type above a given threshold.  The agent can only reveal private skill-level to the principal indirectly, by combining it with an amount of effort into a publicly displayable signal of quality.  However, the agent is limited by a budget on effort, which induces a key difficulty for the principal: amongst agent types exhibiting similar-quality signals, how to distinguish between talented agents with high-skill-low-budget types and endowed agents with low-skill-high-budget types, while contending with agents' incentive-compatability constraints.

A key observation from the model is that it may be beneficial to admit students with average test scores with nonzero probability, while at the same time always admitting students with the highest test scores.
By not guaranteeing admission for students with average test scores, we limit the incentive for those students to put in effort.
High-skill agents (regardless of their budget) find effort less costly than low-skill students; therefore, as we decrease the probability of admission, the low-skill students will reduce their effort more sharply than high-skill students.
Loosely speaking, if we lower the probability of admission enough for students with average test scores, the equilibrium level of effort will drop until the high-skill disadvantaged students' budget constraint is no longer binding.
This allows the school to screen efficiently, at the cost of admitting high-skill students at a lower rate.

As a result, these randomized admission policies make it possible to implement a counter-intuitive outcome.
A student with high skill but low maximum test score (due to limited budget) can receive strictly larger allocation than a student with low skill but high maximum test score.
The latter student is able to achieve scores that are strictly higher than the former student can achieve, but the benefit of obtaining those scores (some probability of admission) is not worth the effort for a low-skill student.

Our main result formalizes this intuition.
It gives (1) a characterization of the structure of the optimal mechanism for a (one-agent) setting with 2 skill types and $n$ budget types, and (2) a polynomial-time algorithm to find it.  An interpretation of our main result is that high-skill agent types may be shown {\em preference} over high-budget types despite the difference in the types' {\em exogenous} resources.  Thus, our setting effectively studies the possibilities and limits of improved-welfare of allocation to effort-budgeted agents.

The paper ends with an introductory study of an extended setting which introduces {\em uniform, exogenous, unconditional subsidies} to relax the agents' budget constraints.\footnote{\label{foot:subsidyunits} Subsidies, measured in units of effort, can for example be monetary transfers from third-parties that increase an agent's effort-budget by freeing up time by reducing other paid work or by buying services.}  Intuitively, the goal is to modify the environment of the admissions problem (as screening) to further increase the balance of allocation in favor of high-skilled types.  Subsidies are a potent intervention because high-skill, budget-constrained agents are best able to use additional effort to increase their highly-valued allocations.  We show that the setting with subsidies has optimal mechanisms with the same characterization as the original screening problem.

\paragraph*{Related Works}

Previous literature has varied its modeling of this central challenge of screening. \citet{sti-75} models agents as having private abilities (types) that the market doesn't observe, and agents with higher abilities have economic incentives to be identified. \citet{SW-81} studied the role of interest rates as a screening device, and showed that returns are not necessarily monotone with respect to interest rates -- a result that holds in equilibrium whenever borrowers \textit{strategically} react to the interest-rate  mechanism. 

In addition to the economics literature on screening, this work contributes to ongoing research on strategic classification, mechanism design with budgeted agents, and fairness.

There is a well-developed literature on mechanism design where agents face budget constraints. Earlier work focused on the case where budgets were public knowledge \citep[e.g.][]{laffont1996optimal,maskin2000auctions}.
More recent work, like ours, focuses on the case where the agents' budgets are their private knowledge \citep{pai2014optimal,GLLLS23,FHL23}.
Typically, budgets are monetary: they represent upper bounds on how much each agent can transfer to the principal. In contrast, we consider budgets on effort: upper bounds on how much effort the agent can put into its task.

In recent years, there has been a lot of interest in strategic classification problems, where a principal is trying to classify agents on the basis of observed scores and agents are able to manipulate (or ``game'') the scores to influence the principal's actions~\citep{Hardt2016,revealed_preferences,hu2019disparate,Milli2018TheSC, Ahmadi2021TheSP,adversarial_games_pred,Frankel2019ImprovingIF,braverman2020role,Kleinberg2018HowDC, harris2021stateful, Alon2020MultiagentEM, xiao2020optimal, Miller2019StrategicCI, Haghtalab2020MaximizingWW, Bechavod2020CausalFD, Ahmadi2022On}.
Our model can be considered a strategic classification problem where the school attempts to classify students into ``admit" or ``not admit," but students are strategic in how much effort they exert.  Expounding on the works most similar to our model: \citet{braverman2020role} show the power of randomization when agents are able to manipulate their scores; and 
\citet{hu2019disparate} study a similar problem where disadvantaged students find it more difficult to manipulate their scores.

In most models of strategic classification, agents obscure their true type at a cost. As a result, costly effort makes scores less informative.
In contrast, in our model of screening, even high skill students need to put in effort in order to achieve a high score (albeit less effort than low-skill students).
If no students put in effort, they will all achieve a score of zero, and the school will not be able to distinguish high-skill from low-skill students.
As a result, costly effort is necessary for scores to be informative in our model.
We must balance the benefits of costly effort in screening with the challenges of costly effort in strategic classification.

Finally, our work relates to a growing literature on fairness in mechanism design and algorithms.
Much of this literature is concerned with fair treatment of different subgroups (e.g., based on demographic variables like race or gender), and various different definitions of fairness have been proposed and criticized \citep[e.g.,][]{corbett2018measure, Kleinberg18}, some of which has been explicitly applied to school admissions \citep[e.g.,][]{KLMR18}.
In line with the fairness literature, we consider the implications of a biased test (where high budget students may perform better, regardless of their skill) for admissions.
Unlike race and gender, the subgroups we are interested in (students with a particular budget) are not publicly observable.
Like \citet{RKLM20} and \citet{JKLPRV20}, we explicitly consider how economic incentives interact with policies designed to correct for sources of unfairness.

\section{Setting and Fundamental Structures}
\label{s:setting}

A principal $\princ$ considers admitting an agent $\agent=(\skill,\budg)$ with private types as skill $\skill$ and budget $\budg$ (budget on {\em effort}, see below).   
The agent's skill and budget are treated as independent, positive Bayesian variables drawn respectively from known distributions $\sdist$ with support $\ssupp\in\reals_+$ and $\bdist$ with support $\bsupp \in [0,1]$, i.e., $\skill\sim \sdist$ and $\budg\sim \bdist$.  The principal only wants to admit the agent in the case that the agent's skill is above a threshold $\thresh\in \reals_+$ (which we implicitly treat as the principal's fixed type).  In summary, the principal's problem is an admission game $\game = (\sdist, \bdist, \thresh)$.

An agent of any skill will want to be admitted and thus, the principal must design a test which uses incentives to elicit information from the agent.  The principal will ask the agent to commit to a {\em private} level of effort $\eff$ which (a) is constrained by individual budget $\budg$, and (b) induces a {\em public}, deterministic signal of quality $\qua = \skill\cdot \eff$. Note that quality is a multiplicative function of effort, rather than additive. This captures two features of our motivating example of school admissions: (i) even high-skill agents that put in no effort will obtain a low score, but (ii) high-skill agents require less effort to achieve a given score than low-skill agents.\footnote{In contrast, suppose quality were an additive function $q=s+e$ of skill and effort. Then property (ii) would hold, but not property (i).}

The principal's problem will be to design an admission allocation rule $\balloc: \reals_+\rightarrow[0,1]$ which maps quality $\qua$ to a stochastic allocation $\alloc$ of admitting the agent.  Practically, the principal's challenge is to optimally discriminate against resource-rich agent types with quality resulting from large effort-budgets, in favor of agents with quality resulting from high skill.  (Note, any ``reasonable" rule will inherently admit all high-skill-high-budget agents.)

The agent's utility is defined to be $-\infty$ if effort exceeds budget, and otherwise is defined to be the probability of allocation minus effort:
\begin{equation}
    \label{eqn:agentutilalloc}
    \util_{\agent}(\eff, \alloc) = \alloc-\eff
\end{equation}
\noindent which implicitly sets the agent's value of being admitted to 1.  Utility can be equivalently written as a function of the allocation rule $\balloc$ and either effort or quality:
\begin{align}
    \label{eqn:agentutilrule}
    \util_{\agent}(\balloc,\eff) &= \balloc(\skill\cdot\eff)-\eff
    & \text{or} && \util_{\agent}(\balloc,\qua) &= \balloc(\qua) - \sfrac{\qua}{\skill}
\end{align}
\noindent (Further, we may drop the input $\balloc$ where its assignment is clear from context.) The agent perceives the allocation rule $\balloc$ as a menu (for which the domain is quality space), albeit top-truncated at the agent's maximum quality set by $\quamax = \skill\cdot\budg$.  This perspective induces for the agent an optimal utility function $\util^*_{\agent}$ and an allocation rule $\alloc$ in skill space (which overloads notation):
\begin{align}
\label{eqn:agentoptvsrule}
    \util^*_{\agent}(\balloc, \skill) &\vcentcolon= \max_{\eff\in[0,\budg]} \balloc(\skill\cdot\eff) - \eff = \max_{\qua\in[0,\quamax]} \balloc(\qua) - \sfrac{\qua}{\skill}\\
\label{eqn:agentoptalloc}
    \alloc = \alloc(\balloc,\skill) &\vcentcolon= \balloc(\skill\cdot \left[\argmax_{e\in[0,\budg]} \balloc(\skill\cdot\eff)-\eff\right])
\intertext{For a given agent $\agent$, the principal's utility from admitting $\agent$ is $\util_{\princ}(\agent~|~\text{admitted})=\skill-\thresh$.  Thus, our principal's mechanism design problem is to maximize $\util_{\princ}(\game, \balloc)$ which is the expected utility from an admitted agent's skill versus the threshold, weighted by allocation probability:}
\label{eqn:princproblem}
    \max_{\balloc} \util_{\princ}(\game, \balloc) &\vcentcolon= \max_{\balloc}  \expecta_{\agent\sim(\sdist\times\bdist)}\left[ \alloc(\balloc,\skill)\cdot\left(\skill-\thresh \right)\right]
\end{align}

\paragraph*{Threshold Mechanisms} A natural mechanism to consider is a {\em threshold mechanism} with the threshold set in quality space.

\begin{definition}
\label{def:threshmech}
A {\em (deterministic) threshold mechanism} $\balloc^{\quath}$ sets a quality threshold $\quath\in\reals_+$ and admits an agent if and only if the agent exhibits public quality $\qua\geq\quath$. 
\end{definition}

\noindent The intuition for a threshold mechanism is that an agent who is able to exhibit the threshold quality with effort less than budget will put in the (minimal amount of) effort necessary to be admitted with probability 1; versus, an agent with maximum-quality $\quamax$ less than the threshold will put in zero effort and get passed.  Recall that the agent's skill and budget are independent in our setting.  The role of thresholds generally is to conditionally allocate agents in decreasing order of skill:

\begin{fact}
\label{fact:highskillfirst}
Given a population of agents as $\sdist\times\bdist$, consider the subset $\bsupp_{\bar{\budg}}$ of agent types which conditionally have a specific budget $\bar{b}$.  For a threshold mechanism with any $\quath>0$, the subset of $\bsupp_{\bar{\budg}}$ of agent skill-types which are admitted is upward-closed.
\end{fact}

\noindent \Cref{fact:highskillfirst} implies that threshold mechanisms are sufficient for the special case in which there is only one budget type (with the proof of \Cref{prop:threshoptsinglebudg} in \Cref{a:proof_thresh}):

\begin{restatable}{proposition}{deterministic}
\label{prop:threshoptsinglebudg}
Assume that an agent has constant budget $\bar{\budg}$ on effort, i.e., the distribution $\bdist$ is a singular point mass.  The threshold mechanism $\balloc^{\quath}$ with $\quath = \thresh\cdot\bar{\budg}$ is optimal.
\end{restatable}

\noindent Intuitively, \Cref{prop:threshoptsinglebudg} holds because single-budget is a simple setting in which quality-thresholds directly implement skill-thresholds, in particular for the principal's threshold $\thresh$.

To outline this section: \Cref{s:suboptofthresh} shows that threshold mechanisms are not optimal for arbitrary distributions $\bdist$ and thus, we will need more-complicated mechanism forms.  \Cref{s:feasanddiscriminate} quantifies agent feasibility to achieve a given quality-allocation pair and, given an allocation rule $\balloc$, discusses implications of feasibility for optimal design.  \Cref{s:graphinterpret} gives geometric interpretation of agent types $(\skill,\budg)$ and their demand under an allocation rule $\balloc$ (for input as quality $\qua$).

\subsection{Generalization of Threshold Mechanisms to Lottery Menus}
\label{s:suboptofthresh}

This section states that deterministic threshold mechanisms are \textit{not optimal} in general (when the distribution over budgets has multiple support). Consequently, we need to generalize the class of mechanisms being considered.  This section gives the sufficient extension to {\em lottery menus} (\Cref{def:lotterymenus}  below).

Insufficiency of deterministic thresholds is stated simply:

\begin{proposition}
\label{prop:threshinsuff} For admission games $\game$ in which the set of budgets is multiple, i.e. $|\bdist|>1$, (deterministic) threshold mechanisms are not optimal generally.
\end{proposition}

\noindent The proof is by counter-example -- we give the details and analysis of \Cref{ex:threshinsuff} in \Cref{a:exthreshinsufficient} where we conclude that all deterministic threshold mechanisms are dominated by stochastic allocation $\alloc=(1-\sfrac{\eps_{\lo}}{1+\eps_{\lo}}-\eps_{\alloc})$ with ``small" $\eps_{\alloc}$ for agents exhibiting at least a minimum quality.

Although we must now consider allocation rules $\balloc$ more generally than threshold mechanisms, without loss of generality, we may assume monotonicity of $\balloc$:

\begin{lemma}[Monotonicity]
\label{lem:monoballoc}
For every admission game $\game$, there exists an optimal allocation rule that is weakly monotone increasing.
\end{lemma}
\begin{proof}
For every allocation rule $\tilde{\balloc}$ that is strictly decreasing somewhere on its domain, the principal gets the same utility from the ``ratcheted" allocation rule $\bar{\balloc}(\tilde{\balloc})$ which increases the allocation in every decreasing region of $\tilde{\balloc}$ to be equal to the left end point of the region, i.e., flat on the region.  (The resulting $\bar{\balloc}(\tilde{\balloc})$ is weakly monotone increasing.)

Principal utility is the same for $\bar{\balloc}$ and $\tilde{\balloc}$ because every agent $\agent = (\skill,\budg)$ gets the same allocation: all qualities $\qua$ where $\bar{\balloc}(\qua)\neq \tilde{\balloc}(\qua)$ are ignored because they are dominated for both functions by the ``ratchet point"-quality (a weakly larger allocation requiring strictly less effort is preferred).
\end{proof}

\noindent Thus, in order to identify the optimal mechanism, we propose lottery menus:\footnote{\label{foot:lott} If we consider admitting multiple agents drawn independently from $\sdist\times\bdist$ and our utility is (independently) additive across decisions, it may be possible to negatively correlate admission decisions to target the total number of admits.  For example, if our setting is discrete and we choose an allocation rule $\balloc$, if $k_{\qua}$ agents apply with the same quality $\qua$, we may decide to run a {\em lottery} which admits exactly $\sfrac{1}{\balloc({\qua})}$ of the agents uniformly at random.}

\begin{definition}[Menu]
\label{def:lotterymenus}
A {\em lottery menu} mechanism is a (weakly) monotone allocation rule $\balloc$ with {\em menu options} $(\qua, \alloc=\balloc(\qua))$, where $\alloc$ is the allocation probability for an agent exhibiting quality $\qua$.
\end{definition}

\subsection{Leveraging Agent Feasibility to Improve Screening}
\label{s:feasanddiscriminate}

\noindent This section formalizes the feasibility for an agent to choose a given menu option.  Subsequently, this section explains how lottery menus effectively leverage feasibility to promote the principal's objective: decreasing allocation necessarily discriminates in favor of higher-skill agents (summarized below as \Cref{prop:x-q-bounds-s}; note, we can already observe this effect working in \Cref{ex:threshinsuff}).

Feasibility is due to (a) the budget constraint, and (b) a non-negative utility requirement:

\begin{definition}
\label{def:feasibility}
Menu option $(\qua, \alloc)$ is {\em feasible} for agent $\agent=(\skill,\budg)$ if: 
\begin{enumerate}
    \item {\em (affordability)} minimal effort $\eff^*= \sfrac{\qua}{\skill}$ (to achieve quality $\qua$) is at most $\budg$, i.e., $\eff^* \leq \budg$; and
    \item {\em (rationality)}  $(\qua,\alloc)$ induces non-negative utility for $\agent$, i.e., $\util_{\agent}(\eff^*,\alloc) = \alloc - \eff^*\geq 0$.
\end{enumerate}
\end{definition}

\begin{fact}
\label{fact:feasibility}
Menu option $(\qua,\alloc)$ is feasible for agent $\agent=(\skill,\budg)$ if and only if $\sfrac{\qua}{\skill}\leq\min\{\budg,\alloc\}$.  Upon choosing this option, $\agent$ achieves utility $\util_{\agent}(\qua) = \alloc - \sfrac{\qua}{\skill}$.
\end{fact}

\noindent \Cref{fact:feasibility} implies that we can use stochastic (partial) allocation to improve the principal's expected utility by discouraging a low-skill agent from applying.  Consider two agents described qualitatively as: high-skill-low-budget ($\agenti[\hi]$) and low-skill-high-budget ($\agenti[\lo]$), where we naturally prefer to admit $\agenti[\hi]$.  Intuitively, we decrease $\alloc$ for a fixed $\bar{\qua}$, we get the following effects: (a) for larger $\budg$, the upper bound on $\sfrac{\bar{\qua}}{\skill}$ is set by $\alloc$ ``sooner" (as it decreases, rather than by budget); and (b) rationality is violated for  $\agenti[\lo]$ before it is violated for $\agenti[\hi]$.  Both effects (a) and (b) threaten $\agenti[\lo]$'s utility.  We state this formally as a {\em ceteris paribus} result, where dependence on feasibility is clear in the proof:

\begin{proposition}[The Lotteries-in-Screening Proposition]
\label{prop:x-q-bounds-s}
For a fixed quality $\bar{\qua}$, decreasing the allocation $\balloc(\bar{\qua})$ when an agent exhibits quality $\bar{\qua}$ increases the lower bound on the skill of agents who feasibly choose $(\bar{\qua},\balloc(\bar{\qua}))$.
\end{proposition}
\begin{proof}
The agent's utility is the difference between allocation and effort: $\balloc(\bar{\qua})-\eff$.  Utility is 
0 for a marginally-skilled agent with skill $\skill^*$ who must put in effort $\eff^*=\balloc(\bar{\qua})$ to achieve quality $\bar{\qua}$.  We also have the abstract definition: $\qua = s\cdot \eff$.  Substituting from the definition, we have $\balloc(\bar{q}) = \sfrac{\bar{\qua}}{\skill^*}$.  The quality $\bar{\qua}$ is fixed, thus decreasing the left-hand side requires increasing the skill threshold $\skill^*$.
\end{proof}

\subsection{Geometric Interpretation}
\label{s:graphinterpret}

This section introduces geometric interpretations of the problem (that will be useful for our analysis of optimal mechanisms).  The first of these visualizations is graphical representation of an agent's feasible allocations.  Regions of feasibility map directly onto a graph of an allocation rule $\balloc$ which has quality space as its domain and allocation as its output. As exhibited in \Cref{fig:oneagentfeas}(Top) 
which gives two graphic examples of these regions, we have the following geometric observations:

\begin{fact}
\label{fact:feasgeom}
An agent $\agent$ with skill $\skill$ (ignoring budget and affordability):
\begin{itemize}
    \item is partially identified by a ray out of the origin with slope $\sfrac{1}{\skill}$; this ray necessarily lower-bounds  $\agent$'s feasible region because this is the {\em zero-utility line}, i.e., points $(\qua,\alloc)$ on this line result in $\agent$ achieving utility of 0;\label{fact:geometry}
    \item who chooses a menu option $(\qua,\alloc)$ -- independent of being rational or not -- will get utility equal to the vertical difference between the chosen allocation $\alloc$ and the height $\qua\cdot(\sfrac{1}{\skill})$ of the zero-utility line at $\qua$ (which directly interprets from definitions: $\util_{\agent}(\qua) = \alloc - \sfrac{\qua}{\skill}$).
\end{itemize}
\end{fact}

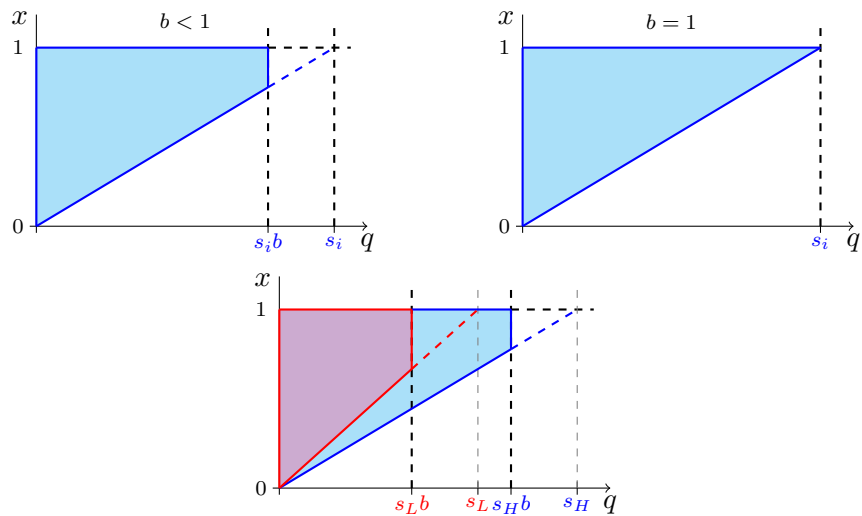
\begin{figure}[t]

\begin{center}
\begin{tikzpicture}[scale = 0.44,pile/.style={->}]

\draw [pile] (-0.2,0) -- (10, 0);
\draw (0, -0.2) -- (0, 6.4);

\fill [Cyan!30] (0,0) -- (7.0,4.2) -- (7.0,5.4) -- (0.0,5.4) -- cycle;

\draw (9.0,-0.2) -- (9.0,0.2);
\draw (7.0,-0.2) -- (7.0,0.2);

\draw [thick, dashed] (7.0,5.4) -- (9.5,5.4);

\draw [thick, dashed] (7.0,0) -- (7.0,6.0);
\draw [thick, dashed] (9.0,0) -- (9.0,6.0);
\draw [blue, thick, dashed] (7.0,4.2) -- (9.0,5.4);

\draw[blue, thick] (0,0) -- (7.0,4.2);
\draw[blue, thick] (7.0,4.2) -- (7.0,5.4);
\draw[blue, thick] (0.0,5.4) -- (7.0,5.4);
\draw[blue, thick] (0.0,0.0) -- (0.0,5.4);

\draw (9.0,-0.5) node {\textcolor{blue}{\figsz $\skilli$}};
\draw (10, -0.5) node {$q$};
\draw (-0.5, 6.3) node {$x$};

\draw (-0.55, 0) node {\figsz $0$};
\draw (-0.55, 5.4) node {\figsz $1$};

\draw (7.0,-0.5) node {\textcolor{blue}{\figsz $\skilli\budg$}};

\draw (4.5,6.2) node {\figszparm $\budg<1$};

\end{tikzpicture}
\hspace{1cm}
\begin{tikzpicture}[scale = 0.44,pile/.style={->}]

\draw [pile] (-0.2,0) -- (10, 0);
\draw (0, -0.2) -- (0, 6.4);

\fill [Cyan!30] (0,0) -- (9.0,5.4)  -- (0.0,5.4) -- cycle;

\draw (9.0,-0.2) -- (9.0,0.2);

\draw [thick, dashed] (9.0,0) -- (9.0,6.0);

\draw[blue, thick] (0,0) -- (9.0,5.4);
\draw[blue, thick] (0.0,5.4) -- (9.0,5.4);
\draw[blue, thick] (0.0,0.0) -- (0.0,5.4);

\draw (9.0,-0.5) node {\textcolor{blue}{\figsz $\skilli$}};
\draw (10, -0.5) node {$q$};
\draw (-0.5, 6.3) node {$x$};

\draw (-0.55, 0) node {\figsz $0$};
\draw (-0.55, 5.4) node {\figsz $1$};

\draw (4.5,6.2) node {\figszparm $\budg=1$};

\end{tikzpicture}
\begin{tikzpicture}[scale = 0.44,pile/.style={->}]

\draw [pile] (-0.2,0) -- (10, 0);
\draw (0, -0.2) -- (0, 6.4);

\fill [Cyan!30] (0,0) -- (7.0,4.2) -- (7.0,5.4) -- (4.0,5.4) -- (4.0,3.6) -- cycle;

\fill [Purple!40] (0,0) -- (4.0, 3.6) -- (4.0,5.4) -- (0,5.4) -- cycle;

\draw (9.0,-0.2) -- (9.0,0.2);
\draw (4.0,-0.2) -- (4.0,0.2);
\draw (6.0,-0.2) -- (6.0,0.2);
\draw (7.0,-0.2) -- (7.0,0.2);

\draw [thick, dashed] (7.0,5.4) -- (9.5,5.4);

\draw [thick, dashed] (4.0,0) -- (4.0,6.0);
\draw [gray, dashed] (6.0,0) -- (6.0,6.0);
\draw [thick, dashed] (7.0,0) -- (7.0,6.0);
\draw [gray, dashed] (9.0,0) -- (9.0,6.0);
\draw [blue, thick, dashed] (7.0,4.2) -- (9.0,5.4);

\draw[blue, thick] (0,0) -- (7.0,4.2);
\draw[blue, thick] (7.0,4.2) -- (7.0,5.4);
\draw[blue, thick] (4.0,5.4) -- (7.0,5.4);
\draw[red, thick] (0,0) -- (4.0,3.6);
\draw[red, thick] (4.0,3.6) -- (4.0,5.4);
\draw[red, thick] (4.0,5.4) -- (0,5.4);
\draw[red, thick] (0,5.4) -- (0,0);
\draw[red, thick, dashed] (4.0,3.6) -- (6.0,5.4);

\draw (9.0,-0.5) node {\figsz \textcolor{blue}{$\skilli[\hi]$}};
\draw (10, -0.5) node {$q$};
\draw (-0.5, 6.3) node {$x$};

\draw (-0.55, 0) node {\figsz $0$};
\draw (-0.55, 5.4) node {\figsz $1$};

\draw (4.0,-0.5) node {\figsz \textcolor{red}{$\skilli[\lo]\budg$}};
\draw (6.0, -0.5) node {\figsz \textcolor{red}{$\skilli[\lo]$}};
\draw (7.0,-0.5) node {\figsz \textcolor{blue}{$\skilli[\hi]\budg$}};

\end{tikzpicture}
\end{center} 
\caption{(Top) A menu option is a point $(\qua,\alloc)$ with coordinates respectively from quality space $\reals_+$ and allocation space $[0,1]$.  The blue regions are feasible for agent $\agent = (\skilli,\budg)$, i.e., $\agent$ can select these menu options (when they exist) and achieve non-negative utility.  The regions' lower-bound line has slope $\sfrac{1}{\skilli}$.  (Bottom) The red region is feasible for an agent $\agenti[\lo]=(\skilli[\lo],\budg)$. The blue region (which entirely encompasses red) is feasible for agent $\agenti[\hi]=(\skilli[\hi],\budg)$. Regarding discussion of \Cref{prop:x-q-bounds-s} in \Cref{s:feasanddiscriminate}, observe how for fixed quality set by $\qua= \skilli[\lo]\cdot\budg$, it is possible to use decreased allocation awarded to a fixed quality (at/below the vertical boundary between red and blue regions), in order to exclusively admit a high-skill agent.}
\label{fig:oneagentfeas}
\label{fig:twoskillfeas}
\end{figure}

\noindent From the points of \Cref{fact:feasgeom}, agent types partitioned by skill $\skilli\in\ssupp$ are identified with their respective zero-utility lines. We illustrate this in \Cref{fig:twoskillfeas}(Bottom) 
by expanding its (Top)
graphics to show a setting with two skill types: low skill $\skilli[\lo]$ and high skill $\skilli[\hi]$.  Within this context, we give formal definitions:

\begin{definition}
\label{def:lowskillline}
The {\em low-skill agents' line} is their zero-utility line with slope $\sfrac{1}{\skilli[\lo]}$ on the (quality, allocation) graph for (budget-unconstrained) low-skill agents.  Similarly, the {\em high-skill agents' line} is their zero-utilily line with slope $\sfrac{1}{\skilli[\hi]}$.  Generally, we refer to zero-utility lines as {\em skill lines}.
\end{definition}

\section{The Optimal Mechanism for 2-skill, Discrete-budget Types}
\label{s:2budget}

This section solves the discrete-type setting for a principal with skill threshold $\thresh$ and a stochastic agent $\agent=(\skill,\budg)$ with type-space defined by two skill-types with $\skilli[\lo]<\thresh<\skilli[\hi]$ and $n$ budget-types with $0<\budgi[1] < \budgi[2] < \ldots < \budgi[n]$.  Due to the discrete type-space, the optimal mechanism may not be unique.  \Cref{thm:2skill_nbudget} is sufficient to identify an optimal mechanism, which is a {\em slanted-stair function}:

\begin{definition}
\label{def:slantedstair}
A {\em slanted-stair function} $f:\reals_+\rightarrow[0,1]$ (as an allocation rule) has $f(0)=0$; and is a weakly increasing function that begins as a sequence of line segments that all have the same (constant), positive derivative.  Each line segment has open lower bound and closed upper bound.  (The function's output must reach 1 and is identically 1 for larger inputs.)

We refer to the line segments as {\em slanted-steps}.  We refer to the (necessarily positive) vertical gaps between slanted-steps as {\em jumps}.
\end{definition}

\noindent For a set of types $T$, let $\Delta(T)$ be the probability simplex over the elements of $T$.  Before giving our main result, we state an interesting observation: there will be nothing in the proof of \Cref{thm:2skill_nbudget} that requires the independence of $\sdist$ and $\bdist$.  Thus to state a stronger main result, we define a {\em correlated} admission game by $\gamecor = (\sdist, \bcorset, \thresh)$ where $\bcorset$ is a set of conditional budget-distributions: one budget-distribution corresponding to each skill-type with positive support in $\sdist$.\footnote{\label{foot:commonbudgsupportwlog} Assuming discrete budget-distributions, while elements of $\bcorset$ may have distinct support, it is without loss of generality to assume that they all have common, enumerated support $\budgi[1],\ldots,\budgi[n]$ because any locally-unused budget type $\budgi$ can be locally assigned probability 0.}

Agents in a correlated game have the same description as in the original, independent game.  By contrast, the principal's objective must be updated to reflect the correlation:

\begin{equation}
    \label{eqn:princproblemcorr}
    \max_{\balloc} \util_{\princ}(\gamecor, \balloc) \vcentcolon= \max_{\balloc}  \expecta_{\agent\sim(\sdist,~\bcorset)}\left[ \alloc(\balloc,\skill)\cdot\left(\skill-\thresh \right)\right]
\end{equation}

\begin{restatable}[Main Result]{theorem}{main}
\label{thm:2skill_nbudget}
Given a correlated admission game $\gamecor = (\sdist = \Delta(\{\skilli[\lo],\skilli[\hi]\}),\\\bcorset=\{\bdist_{\skilli[\lo]}=\Delta_{\skilli[\lo]}(\{\budgi[1],\ldots,\budgi[n]\}),\bdist_{\skilli[\hi]}=\Delta_{\skilli[\hi]}(\{\budgi[1],\ldots,\budgi[n]\})\},\thresh)$ with $0<\skilli[\lo]<\thresh<\skilli[\hi]$ and $0<\budgi[1] < \budgi[2] < \ldots < \budgi[n]$.  There exists an optimal mechanism $\balloc^*$ for the correlated admission game $\gamecor$ that is a slanted-stair function $f$ with constant slope equal to $\sfrac{1}{\skilli[\lo]}$ and at most one jump, and with:
\begin{enumerate}
    \item the region of the first slanted-step characterized by: equality to the {\em low-skill agents' zero-utility line};
    \item the quality-index at which $f$ jumps $\qua^{\text{{\em jump}}}$ (if it exists) characterized by: occurring either at quality $\quai[0]=0$, or occurring at some maximum-possible quality exhibited by some low-skilled agent $\agenti[\lo,i]=(\skilli[\lo],\budgi)$, i.e., at some $\qua^{\text{{\em jump}}} = \quamaxi[\lo,i]=\skilli[\lo]\cdot\budgi$;
    \item the region of the second slanted-step (if it exists) characterized by: the quality at which $f$ intersects the allocation-of-1 horizontal line is the maximum-possible quality exhibited by some high-skilled agent $\agenti[\hi,j]=(\skilli[\hi],\budgi[j])$, i.e., at some $\qua^{\alloc=1} = \quamaxi[\hi,j]=\skilli[\hi]\cdot\budgi[j]$. 
\end{enumerate}
\end{restatable}

\noindent (Note, optimal assignment of mechanism parameters and the given characterizations of \Cref{thm:2skill_nbudget} are sufficient to identify the height of the vertical jump, starting from the low-skill agents' line.)

The proof of \Cref{thm:2skill_nbudget} depends on a sequence of lemmas which we state at the end of this section.  The proofs of \Cref{thm:2skill_nbudget} and its supporting lemmas appear in the main version of the paper. 
 Graphically, the optimal menu (which may be discrete, corresponding to our discrete setting) will qualitatively have the single-jump structure of \Cref{fig:opt_2xn}(Top) 
with menu options on only two line segments (as two slanted-steps).  Multi-jump structures are precluded, such as the three-slanted-steps in \Cref{fig:subopt_2xn}(Bottom). 

\begin{figure}[t]
\begin{center}
\begin{tikzpicture}[scale = 0.44,pile/.style={->}]

\draw [pile] (-0.2,0) -- (15.5, 0);
\draw (0, -0.2) -- (0, 6.4);

\draw (13.50,-0.2) -- (13.50,0.2);
\draw (5.4,-0.2) -- (5.4,0.2);
\draw (9.0,-0.2) -- (9.0,0.2);
\draw (10.8,-0.2) -- (10.8,0.2);

\draw [dashed] (0.0,5.4) -- (14.5,5.4);

\draw [gray, dashed] (5.4,0) -- (5.4,6.0);
\draw [gray, dashed] (9.0,0) -- (9.0,6.0);
\draw [gray, dashed] (10.8,0) -- (10.8,6.0);
\draw [gray, dashed] (13.5,0) -- (13.5,6.0);

\draw[blue, dashed] (0,0) -- (13.5,5.4);
\draw[red, dashed] (0.0,0.0) -- (10.8,5.4);
\draw[ForestGreen, thick] (0,-0.04) -- (5.4,2.66);
\draw[ForestGreen, thick] (5.4,3.6) -- (9.0, 5.4);

\draw[fill=black] (0,0) circle (0.12cm);
\draw[fill=black] (2.7,1.31) circle (0.12cm);
\draw[fill=black] (3.6,1.76) circle (0.12cm);
\draw[fill=black] (5.4,2.66) circle (0.12cm);

\draw[fill=white] (5.4, 3.6) circle (0.12cm);
\draw[fill=black] (6.3,4.05) circle (0.12cm);
\draw[fill=black] (8.1,4.95) circle (0.12cm);
\draw[fill=black] (9.0,5.4) circle (0.12cm);

\draw (13.5,-0.5) node {\figsz \textcolor{blue}{$\skilli[\hi] $}};
\draw (15.5, -0.5) node {$q$};
\draw (-0.5, 6.3) node {$x$};

\draw (-0.55, 0) node {\figsz $0$};
\draw (-0.55, 5.4) node {\figsz $1$};

\draw (5.4,-0.5) node {\figsz \textcolor{red}{$\skilli[\lo]\budgi$}};
\draw (10.8, -0.5) node {\figsz \textcolor{red}{$\skilli[\lo]$}};
\draw (9.0,-0.5) node {\figsz \textcolor{blue}{$\skilli[\hi]\budgi[j]$}};

\end{tikzpicture}
\hspace{0.6cm}
\begin{tikzpicture}[scale = 0.44,pile/.style={->}]

\draw [pile] (-0.2,0) -- (15.5, 0);
\draw (0, -0.2) -- (0, 6.4);

\draw (13.50,-0.2) -- (13.50,0.2);
\draw (5.4,-0.2) -- (5.4,0.2);
\draw (9.0,-0.2) -- (9.0,0.2);
\draw (10.8,-0.2) -- (10.8,0.2);

\draw [dashed] (0.0,5.4) -- (14.5,5.4);

\draw [gray, dashed] (5.4,0) -- (5.4,6.0);
\draw [gray, dashed] (9.0,0) -- (9.0,6.0);
\draw [gray, dashed] (10.8,0) -- (10.8,6.0);
\draw [gray, dashed] (13.5,0) -- (13.5,6.0);

\draw[blue, dashed] (0,0) -- (13.5,5.4);
\draw[red, dashed] (0.0,0.0) -- (10.8,5.4);
\draw[ForestGreen, thick] (0,-0.04) -- (5.4,2.66);
\draw[ForestGreen, thick] (5.4,3.25) -- (8.1, 4.6);
\draw[ForestGreen, thick] (8.1, 4.95) -- (9.0, 5.4);

\draw[fill=black] (0,0) circle (0.12cm);
\draw[fill=black] (2.7,1.31) circle (0.12cm);
\draw[fill=black] (3.6,1.76) circle (0.12cm);
\draw[fill=black] (5.4,2.66) circle (0.12cm);

\draw[fill=white] (5.4, 3.25) circle (0.12cm);
\draw[fill=black] (6.3,3.7) circle (0.12cm);
\draw[fill=black] (8.1,4.6) circle (0.12cm);
\draw[fill=white] (8.1,4.95) circle (0.12cm);
\draw[fill=black] (9.0,5.4) circle (0.12cm);

\draw (13.5,-0.5) node {\figsz \textcolor{blue}{$\skilli[\hi] $}};
\draw (15.5, -0.5) node {$q$};
\draw (-0.5, 6.3) node {$x$};

\draw (-0.55, 0) node {\figsz $0$};
\draw (-0.55, 5.4) node {\figsz $1$};

\draw (5.4,-0.5) node {\figsz \textcolor{red}{$\skilli[\lo]\budgi$}};
\draw (10.8, -0.5) node {\figsz \textcolor{red}{$\skilli[\lo]$}};
\draw (9.0,-0.5) node {\figsz \textcolor{blue}{$\skilli[\hi]\budgi[j]$}};

\end{tikzpicture}
\end{center} 
\caption{(Top) A one-jump, slanted-stair allocation curve $\balloc$ (solid green) with $\qua^{\text{jump}}=\skilli[\lo]\cdot\budgi$ and $\qua^{\alloc=1} = \skilli[\hi]\cdot\budgi[j]$. The black dots are an example of discrete menu options. 
Recall that agent utility is interpretable as the vertical difference between allocation and (zero-utility) skill line.  Any low-skilled agent $\agenti[\lo]=(\skilli[\lo],\budgi[k])$ with $\quamaxi[\lo,k]=\skilli[\lo]\cdot\budgi[k]\leq \skilli[\lo]\cdot\budgi$ will choose menu option $(0,0)$ (per the tie-breaking rule, see \Cref{def:tiebreak}).  Any low-skilled agent with $\quamaxi[\lo,k]=\skilli[\lo]\cdot\budgi[k]> \skilli[\lo]\cdot\budgi$ will choose $(\skilli[\lo]\budgi+\epsilon,\balloc(\skilli[\lo]\budgi+\epsilon))$ with $\epsilon\rightarrow0$. Each high-skilled agent $\agenti[\hi]=(\skilli[\hi],\budgi[k])$ with $k < j$ will achieve maximum quality $\quamaxi[\hi,k]=\skilli[\hi]\cdot\budgi[k]<\skilli[\hi]\budgi[j]$; and those with $k\geq j$ will achieve quality $\skilli[\hi]\budgi[j]$ (and are allocated with probability 1).  (Bottom) A two-jump, slanted-stair allocation curve $\balloc$ (solid green).}
\label{fig:opt_2xn}
\label{fig:subopt_2xn}
\end{figure}
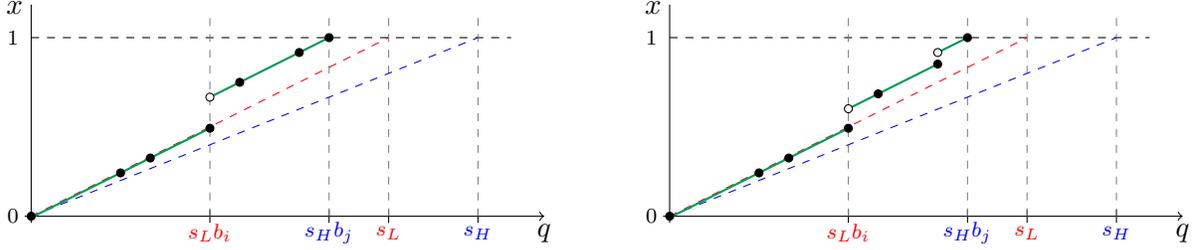

The statement of \Cref{thm:2skill_nbudget} induces the following corollary regarding the polynomial running time of a brute-force algorithm that searches over the possible combinations of jump-points and jump-heights, which is sufficient to find the optimal algorithm of the statement's setting.

\begin{corollary}[Running Time]
\label{cor:2skill_nbudget_runtime}
Given a correlated admission game $\gamecor = (\sdist = \Delta(\{\skilli[\lo],\skilli[\hi]\}),\\\bcorset=\{\bdist_{\skilli[\lo]}=\Delta_{\skilli[\lo]}(\{\budgi[1],\ldots,\budgi[n]\}),\bdist_{\skilli[\hi]}=\Delta_{\skilli[\hi]}(\{\budgi[1],\ldots,\budgi[n]\})\},\thresh)$ with $0<\skilli[\lo]<\thresh<\skilli[\hi]$ and $0<\budgi[1] < \budgi[2] < \ldots < \budgi[n]$ and -- per \Cref{thm:2skill_nbudget} -- the sufficient, discrete search space for an optimal algorithm.

The optimal mechanism may be identified by a brute-force search over the $O(n^2)$ unknown combinations of parameters of the optimal characterization (\Cref{thm:2skill_nbudget}). The time to evaluate each allocation rule (resulting from a combination of parameters) also runs in polynomial time.
\end{corollary}

\subsection{Discussion of Optimal Characterization in \Cref{thm:2skill_nbudget}}
\label{s:discussmainthm}

Having a characterization of optimal mechanisms, we would like to understand qualitatively their performance.  We will discuss two dimensions of efficacy: (1) mechanism performance, of course, as the originally-defined objective; and (2) {\em fairness}, which informally is a measurement of how well outcomes-per-agent-type conform to some definition of what outcomes the agents {\em arguably deserve}, specifically compared to other agents' type-outcome pairs.

Regarding mechanism performance, we know that the single-jump, slanted-stair characterization of \Cref{thm:2skill_nbudget} improves on the (deterministic) threshold mechanisms of \Cref{def:threshmech} which are not generally optimal (by \Cref{prop:threshinsuff}), except for games with convenient distributions $\sdist$ and $\bdist$ (e.g., \Cref{prop:threshoptsinglebudg}).  On the other hand, optimal mechanism performance still falls short of the {\em offline optimal} benchmark 
which has full information and which is generally unachievable; rather, we may use it as a reference mechanism to which we compare performance:

\begin{definition}
\label{def:offopt}
Given a stochastic agent $\agent=(\skill,\budg)$, the {\em offline optimal mechanism} for a principal requiring skill-threshold $\thresh$ -- which is assumed to know the realized skill type of the agent as $\hat{\skill}\sim\sdist$ -- admits the agent if $\hat{\skill} > \thresh$ and only if $\hat{\skill}\geq \thresh$; and this admission decision is independent of the agent's realized budget type $\hat{\budg}\sim\bdist$.
\end{definition}

\noindent The offline optimal mechanism is unconditionally optimal, as it fully allocates every agent with skill above the threshold and fully rejects every agent with skill below it.  In order to increase the performance of mechanisms beyond what is possible from \Cref{thm:2skill_nbudget} -- i.e., from standard mechanism design subject to agents' incentive compatibility constraints -- in \Cref{s:subsidy} we consider a modified admission problem in which the agent may have exogenous access to a {\em subsidy}.

Regarding fairness, we first must consider the philosophical concept of what comparisons between distinct agents' type-outcome pairs may arise as fair or as unfair within the parameters of our model (\Cref{s:setting}).  Loosely summarizing: our agents independently have higher or lower skills and higher or lower budgets; and by best-responding to a given allocation rule based on skill and budget, agents are consequently admitted with larger or smaller probability.  Reasonably, agent ``skill" is positively correlated with student value and agent budget is independent, so we posit that higher skill types are {\em more-deserving} of being admitted than lower types, 
independent of budget.  Moreover, the degree of worthiness should increase with increasing {\em cardinal} difference in skill types.

Thus, we consider the following concept of fairness: regardless of budget, larger (admission) allocations given to lower-skilled agents are comparatively judged to be unfair outcomes as the higher-skilled type is more-deserving; and the larger the skill-difference, the larger the unfairness.  Furthermore the strict contrapositive also holds: comparatively larger allocations given to higher-skilled agents are more fair.  However, the choice of function used to measure technically the unfairness of an allocation rule remains debatable.

From the following intuition, the mechanism design problem of our admission-game model should be positively aligned with objectives resulting from our concept of fairness.  First, recall the principal's utility from admitting an agent $\agent$, which is $\util_{\princ}(\agent~|~\text{admitted})=\skill-\thresh$.  Given this utility function, the principal has a precise, cardinal utility measure over admitting agent skill-types, which has both order and cardinality aligned with fairness as desired, regardless of the technical fairness measure. 
 I.e., the principal is incentivized to choose an allocation rule that increases fairness.  In at least one sense, this is strictly true, which moreover motivates the principal's objective function itself (see equation \eqref{eqn:princproblemcorr}) as a formal example of fairness measure:
 
 \begin{fact}
 \label{fact:princfairaligned}
 Where incentive compatibility permits, the principal is incentivized to inherently prefer that between two agent types with different skill levels, the agent type with higher skill will receive the larger allocation probability.
\end{fact}

\begin{restatable}{corollary}{princasfairmeasure}
\label{cor:princ-as-fairness-dominates} Given a 
correlated admission game $\gamecor = (\sdist = \Delta(\{\skilli[\lo],\skilli[\hi]\}),\\\bcorset=\{\bdist_{\skilli[\lo]},\bdist_{\skilli[\hi]}\},\thresh)$ with $0<\skilli[\lo]<\thresh<\skilli[\hi]$ and $0<\budgi[1] < \budgi[2] < \ldots < \budgi[n]$.  For the fair mechanism design problem which maximizes the fairness measure set equal to the principal's utility function, the optimal mechanism and characterization of optimal mechanisms are determined identically to \Cref{thm:2skill_nbudget}.
\end{restatable}

\noindent  Second, the offline optimal mechanism can illustrate the alignment between the principal's mechanism design incentives and fairness.  On one hand, offline optimal represents perfect -- albeit generally unachievable -- performance for the mechanism.  On the other hand, by giving allocation 1 to an upward-closed set of skill-types above $\thresh$, allocation 0 to a downward-closed set of skill-types below $\thresh$, and any constant allocation to skill-type exactly $\thresh$, the allocation is arguably fair because no rejected skill-type can protest for increased allocation on the basis that it is strictly more-deserving than any admitted skill-type.  Thus, the offline optimal mechanism as ideal-objective further aligns the principal and fairness.

For purposes of space, we defer discussion of a third intuitive perspective supporting the alignment of optimal mechanisms and fairness to \Cref{a:compare-ss-to-dt}.

\subsection{The Proof of \Cref{thm:2skill_nbudget}}
\label{s:proofoutlinemainthm}

We need one more critical detail to set up the proof of \Cref{thm:2skill_nbudget}.  Depending on allocation rule $\balloc$, an agent $\agent$ may be indifferent between a set of quality-allocation menu options that are optimal for $\agent$.  To address this, we define our tie-breaking rule:

\begin{definition}
\label{def:tiebreak}
When an agent's set of optimal menu options is multiple, the {\em tie-breaking rule} is: 
all agents choose {\em the smallest menu option of the set}.  (Note, ``smallest" is the same in either dimension of quality or allocation.)
\end{definition}

\noindent This tie-breaking rule is material for our results: it is sufficient to break ties optimally in favor of the principal's objective.\footnote{\label{foot:tiebreakex} This tie-breaking rule is justified similarly to tie-breaking in other areas of mechanism design, e.g., in auctions with a revenue objective in which agents with value equal to price are assumed to buy, in favor of the designer's objective.  Intuitively, the justification is that small perturbations to the design can achieve the same outcome within arbitrary (lossy) required precision; so instead, we simplify the analysis by allowing ties and breaking them favorably, rather than accounting for a notation-heavy perturbation.}  Recalling that utility is equal to the vertical difference between the allocation and the height of the zero-utility line (\Cref{fact:feasgeom}), the key effect of tie-breaking is observed in \Cref{fig:opt_2xn}: within a region of a single slanted-step, {\em low-skill agents are indifferent everywhere and choose the minimal allocation at the left endpoint of the region}.  This tie-breaking rule applies for all result statements and proofs in this paper.

As an overview, the proof of \Cref{thm:2skill_nbudget} proceeds as a search for the optimal mechanism.  This search is organized as a sequence of reductions of the search space: it starts with an allocation rule that is monotone (\Cref{lem:monoballoc} on page~\pageref{lem:monoballoc}) but is otherwise arbitrary; and then with each successive lemma, we prove that it is sufficient to restrict attention to a smaller set of allocation rules.  \Cref{lem:lemmaB} is the last reduction in the sequence and states that the optimal mechanism must be a slanted-stair function (\Cref{def:slantedstair}) with at most one jump.  The final proof of \Cref{thm:2skill_nbudget} starts from the statement of \Cref{lem:lemmaB} and proves the additional details in its own statement.

All of the following lemmas assume the same setting as the statement of \Cref{thm:2skill_nbudget}, which is: given a correlated admission game $\gamecor = (\sdist = \Delta(\{\skilli[\lo],\skilli[\hi]\}), \bcorset=\{\bdist_{\skilli[\lo]}=\Delta_{\skilli[\lo]}(\{\budgi[1],\ldots,\budgi[n]\}),\bdist_{\skilli[\hi]}=\Delta_{\skilli[\hi]}(\{\budgi[1],\ldots,\budgi[n]\})\},\thresh)$ with $0<\skilli[\lo]<\thresh<\skilli[\hi]$ and $0<\budgi[1] < \budgi[2] < \ldots < \budgi[n]$.

With this overview in place, the sequence of reductions of the search space is:

\begin{restatable}[Lower bound]{lemma}{lemmaX}
\label{lem:lemmaX}
An optimal allocation rule is never under the low-skill agents' line.
\end{restatable}

\begin{restatable}[Strong monotonicity]{lemma}{lemmaY}
\label{lem:lemmaY}
There exists an optimal allocation rule $\balloc^*$ that everywhere has a derivative lower bound set by $\sfrac{1}{\skilli[\lo]}$ (the slope of the low-skill agents' line). 
\end{restatable}
 
\begin{restatable}[Constant allocation slope]{lemma}{lemmaA}
\label{lem:lemmaA}
There exists an optimal allocation rule $\balloc^*$ that is a slanted-stair function, i.e., it everywhere has constant derivative equal to $\sfrac{1}{\skilli[\lo]}$ (the slope of a low-skill agents' line), allowing for arbitrary, discretely-indexed, positive, vertical jumps.
\end{restatable}

\begin{restatable}[A corner-case exclusion]{lemma}{lemmaZ}
\label{lem:lemmaZ}
There exists an optimal allocation rule $\balloc^*$ for which the optimal menu option of the agent type with smallest maximum-quality gives 0-allocation.  (This agent is $\agent= (\skilli[\lo],\budgi[1])$ with $\quamaxi[\lo,1] = \skilli[\lo]\cdot\budgi[1]$.)
\end{restatable}

\begin{restatable}[Sufficiency of at-most one jump]{lemma}{lemmaB}
\label{lem:lemmaB}
There exists an optimal allocation rule $\balloc^*$ that is a slanted-stair function with at most one jump; furthermore, if there is a jump in a given $\balloc^*$, then its allocation in the region of the first slanted-step must be equal to the low-skilled agents' line.
\end{restatable}

\noindent The proofs for each lemma in this sequence appear in the full version of the paper.

{\section{Mechanisms for Agents with Subsidized Effort}
\label{s:subsidy}

This section considers agent subsidies directly in {\em effort-space}.  A budget on effort implies a time-constraint.  Effort-subsidies are an intervention that increases the agent's effort-budget by freeing up an agent's time spent on other obligatory activities.  Technically, we consider subsidies as uniform, additive increases to agents' budget constraints.  These subsidies are offered {\em unconditionally}: agents may spend the time on an outside-option (leisure) activity; or they may invest the time in effort, which they experience as {\em costly} (i.e., as the opportunity cost of the forfeited leisure time).  E.g., subsidies may be provided by performing time-costly tasks for agents' benefit (like uniformly offering free postal pickup/delivery) -- freed from the burden of the task, agents enjoy leisure or spend their time exerting effort in our model.

The main goal of this section is to solve for the characterization of the optimal mechanism of the (modified) admission game which has expanded setting parameters that make it possible to consider a combined-question of screening and {\em design of unconditional subsidies}.  \Cref{cor:2skill_nbudget-subsy} states that its characterization is the same as \Cref{thm:2skill_nbudget}.  We also show that this subsidies setting can only help the principal's objective (in \Cref{prop:possubsyimprovesprinc}).

\subsection{The Setting with Subsidies}
\label{s:subsidysetting}

We add the following elements to the correlated setting of \Cref{thm:2skill_nbudget} (based on \Cref{s:setting}).

The mechanism designer may a priori offer to the agent $\agent = (\skill,\budg)$ an {\em effort-budget subsidy} $\subsy$ from a non-negative range, i.e., the subsidy is $\subsy\in[\subsylo,\subsyhi]$.  The agent accepts the whole subsidy unconditionally and the agent's new budget is $\budg+\subsy$.

It is not possible to restrict access to the subsidy to sub-classes of agent-types: not to high-skill agents and not to disadvantaged agents. 
 The constant subsidy amount is necessarily available to each type indiscriminately because the realizations of an agent's skill/budget types are unknown at the time of the offer, i.e., at the time of subsidized-mechanism design.  While we can not use uniform subsidies to discriminate directly, we will be able to improve the {\em principal's objective} using the following  observation: given an optimal single-jump, slanted-stair allocation (as characterized by \Cref{thm:2skill_nbudget}), note that the budget constraint binds for {\em all} high-skill agents receiving allocation less than 1 and they would benefit from relaxing the budget constraint; but for almost all low-skill agents, the budget constraint is not binding because their utility is constant on each slanted-step.  This first-order-condition analysis suggests that high-skill agents will voluntarily convert unconditional subsidies to effort and increased allocation, whereas low-skill agents will not.

The subsidy (to increase effort-budget) is exogenous as if enacted and paid by an unrelated third party at no cost to the mechanism.  E.g., in an admission problem, the school may be a city's unique, public, magnet high school.  The subsidy may be paid uniformly to each eligible applicant by a citywide scholarship program which is separate from the school's admissions office but which has the money to provide the subsidy (up to $\subsyhi$ per student) and {\em must support} a citywide goal of maximizing utility from specifically the magnet school's admissions policies.  In this case, the magnet school admissions office (as our model's principal) optimizes $\subsy\in[\subsylo,\subsyhi]$ and the scholarship program must approve it.

For this \Cref{s:subsidy}, the updated correlated admission game with subsidies is given by $\gamedy=(\sdist,\bcorset,\thresh,\subsylo,\subsyhi)$.  For a given subsidy $\subsy>0$, agent $\agent=(\skill,\budg)$ has maximum quality $\quamdy=\skill\cdot(\budg+\subsy)$, which is larger than the maximum quality without the subsidy ($\quamax=\skill\cdot\budg$).  The agent's updated optimal utility function $\utilv^*_{\agent}$ and updated optimal allocation rule $\allocdy$ in skill space -- subject to allocation rule $\balloc$ -- are:
\begin{align}
\label{eqn:agentoptvsrulesubsy}
    \utilv^*_{\agent}(\balloc, \skill, \subsy) &\vcentcolon= \max_{\eff\in[0,\budg+\subsy]} \balloc(\skill\cdot\eff) - \eff = \max_{\qua\in[0,\quamdy]} \balloc(\qua) - \sfrac{\qua}{\skill}\\
\label{eqn:agentoptallocsubsy}
    \alloc = \allocdy(\balloc,\skill,\subsy) &\vcentcolon= \balloc(\skill\cdot \left[\argmax_{e\in[0,\budg+\subsy]} \balloc(\skill\cdot\eff)-\eff\right])
\end{align}

\noindent In equation~\eqref{eqn:agentoptallocsubsy}, note that because the subsidy is unconditional, the agent pays the full cost of effort $\eff$, including the (opportunity) cost of effort above the original budget $\budg$.

For a given agent $\agent$, the principal's utility from admitting $\agent$ \underline{remains the function} $\util_{\princ}(\agent~|~\text{admitted})=\skill-\thresh$.  Thus, the principal's updated mechanism design problem is to maximize $\utilv_{\princ}(\gamedy, \balloc, \subsy)$ which is the expected utility from an admitted agent's skill versus the principal's threshold $\thresh$, weighted by allocation probability according to $\allocdy$ (which accounts for the subsidy):

\begin{equation}
\label{eqn:princproblemsubsy}
    \max_{\balloc,~\subsy\in[\subsylo,\subsyhi]} \utilv_{\princ}(\gamedy, \balloc, \subsy) \vcentcolon= \max_{\balloc}  \expecta_{\agent\sim(\sdist,~\bcorset)}\left[ \allocdy(\balloc,\skill,\subsy)\cdot\left(\skill-\thresh \right)\right]
\end{equation}

\subsection{Results with Subsidies}
\label{s:subsidyresults}

The main result of this section is: the optimal mechanism when agents have access to unconditional subsidies has the same characterization as the original game, as described in \Cref{thm:2skill_nbudget}.  Moreover, we are immediately ready to state and prove it as a corollary:

\begin{corollary}
\label{cor:2skill_nbudget-subsy}
Given a correlated admission game with subsidies $\gamedy=(\sdist = \Delta(\{\skilli[\lo],\skilli[\hi]\}),\\\bcorset=\{\bdist_{\skilli[\lo]}=\Delta_{\skilli[\lo]}(\{\budgi[1],\ldots,\budgi[n]\}),\bdist_{\skilli[\hi]}=\Delta_{\skilli[\hi]}(\{\budgi[1],\ldots,\budgi[n]\})\},\thresh,\subsylo,\subsyhi)$ with $0<\skilli[\lo]<\thresh<\skilli[\hi]$ and $0<\budgi[1] < \budgi[2] < \ldots < \budgi[n]$.  The structure of the optimal mechanism has the same characterization as the standard game, as given in \Cref{thm:2skill_nbudget}.
\end{corollary}
\begin{proof}
As part of identifying the optimal mechanism -- according to equation~\eqref{eqn:princproblemsubsy} -- the designer selects an optimal assignment to the subsidy variable $\subsy\in[\subsylo,\subsyhi]$.

Consider an optimal assignment $\subsy^*$ (any element of the $\argmax$ is fine).  The optimal mechanism associated with $\subsy^*$ must be the same as the optimal mechanism for an alternative game $\gamedy'$ which sets parameters $\sdist, \bcorset, \thresh$ to be the same as $\gamedy$, but which assigns the endpoints of allowable subsidies to both be $\subsy^*$, i.e., $\gamedy'$ has $\subsylo=\subsyhi=\subsy^*$.

This corollary then follows directly from \Cref{lem:subsidystatements}(2) below.
\end{proof}

\noindent While \Cref{cor:2skill_nbudget-subsy} is sufficient to give us characterization, it does not give us an algorithm to find the optimal mechanism because it uses theoretical existence of the optimal subsidy $\subsy^*$ without identifying it.

The following observations regarding correlated admission games with subsidies are straightforward.  Omitted proofs in this section appear in the full version of the paper.
\begin{restatable}{lemma}{subsystatements}
\label{lem:subsidystatements}
A correlated admission game with subsidies is $\gamedy=(\sdist = \Delta(\{\skilli[\lo],\skilli[\hi]\}), \bcorset=\{\bdist_{\skilli[\lo]}=\Delta_{\skilli[\lo]}(\{\budgi[1],\ldots,\budgi[n]\}),\bdist_{\skilli[\hi]}=\Delta_{\skilli[\hi]}(\{\budgi[1],\ldots,\budgi[n]\})\},\thresh,\subsylo,\subsyhi)$. 
 Consider arbitrary $\gamedy$, i.e., consider its inputs as variables.
\begin{enumerate}
    \item Without loss of generality, we may reduce $\gamedy$ to a correlated game $\gamedy'$ which has $\subsylo'=0$.
    \item A game $\gamedy$ fixing an exact subsidy by setting $\subsylo=\subsyhi$ is equivalently described by a game $\gamecor_{\gamedy}$ and thus is characterized by the statement of \Cref{thm:2skill_nbudget}.
    \item If $\subsylo=0$, then expanding the original correlated admission game $\gamecor = (\sdist, \bcorset, \thresh)$ to consider admissions with subsidies -- formulated as the updated game $\gamedy$ -- can only increase the utility of the principal.
    \item Given $\sdist, \bcorset,\thresh$, 
    there exists a minimal subsidy upper bound $\subsyhi^m$ such that for all $\subsyhi\geq \subsyhi^m$, the optimal mechanism achieves the offline optimal performance (see \Cref{def:offopt}), i.e., it is able to perfectly discriminate between high-skill and low-skill agent types regardless of their budgets.
\end{enumerate}
\end{restatable}

\noindent \Cref{lem:subsidystatements}(3) is fairly obvious: if $\subsylo=0$, then the principal has the option of ``free disposal" of the subsidy-variable and can do no worse than the game without subsidies.  The more interesting statement is that the principal's objective can only improve for $\subsylo>0$ generally:

\begin{restatable}{proposition}{positivesubsyimprovesprinc}
\label{prop:possubsyimprovesprinc}
    For arbitrary $\subsylo\geq0$, expanding the original correlated admission game $\gamecor = (\sdist, \bcorset, \thresh)$ to consider admissions with subsidies can only increase the utility of the principal.
\end{restatable}

\noindent In the proof of \Cref{prop:possubsyimprovesprinc}, we consider specifically the subsidy $d=\subsylo>0$ and (deterministically) transform the optimal allocation rule without subsidies into a new allocation rule with weakly larger performance given the uniform, unconditional agent's budget-subsidy $\subsylo$.  In particular in comparison to the optimal allocation without subsidies, the new allocation gives all low-skill agent-types weakly smaller allocation, and gives all high-skill agent-types weakly larger allocation.

This new allocation rule is not necessarily optimal for its (subsidized) setting, but by dominating the original setting, its existence proves that the principal's objective can only improve.  On the other hand, the new allocation rule may harm the agents' utilities (for any agent type, except low-skill-low-budget agents who already get 0-allocation before subsidies and who still get 0).  While this assessment is not a final judgment (because the new allocation is not necessarily optimal), it is consistent with observations in \citet{hu2019disparate} which showed that subsidies for disadvantaged agents might harm their utilities.

\bibliographystyle{apalike}

\bibliography{forc23_bib_paper06}

\begin{thebibliography}{}

\bibitem[Ahmadi et~al., 2021]{Ahmadi2021TheSP}
Ahmadi, S., Beyhaghi, H., Blum, A., and Naggita, K. (2021).
\newblock The strategic perceptron.
\newblock In Bir{\'{o}}, P., Chawla, S., and Echenique, F., editors, {\em {EC}
  '21: The 22nd {ACM} Conference on Economics and Computation, Budapest,
  Hungary, July 18-23, 2021}, pages 6--25. {ACM}.

\bibitem[Ahmadi et~al., 2022]{Ahmadi2022On}
Ahmadi, S., Beyhaghi, H., Blum, A., and Naggita, K. (2022).
\newblock On classification of strategic agents who can both game and improve.
\newblock In Celis, L.~E., editor, {\em 3rd Symposium on Foundations of
  Responsible Computing, {FORC} 2022, June 6-8, 2022, Cambridge, MA, {USA}},
  volume 218 of {\em LIPIcs}, pages 3:1--3:22. Schloss Dagstuhl -
  Leibniz-Zentrum f{\"{u}}r Informatik.

\bibitem[Alon et~al., 2020]{Alon2020MultiagentEM}
Alon, T., Dobson, M., Procaccia, A., Talgam-Cohen, I., and Tucker-Foltz, J.
  (2020).
\newblock Multiagent evaluation mechanisms.
\newblock {\em In Proceedings of the AAAI Conference on Artificial
  Intelligence}, 34(02):1774--1781.

\bibitem[Bechavod et~al., 2020]{Bechavod2020CausalFD}
Bechavod, Y., Ligett, K., Wu, Z.~S., and Ziani, J. (2020).
\newblock Causal feature discovery through strategic modification.
\newblock {\em ArXiv}, abs/2002.07024.

\bibitem[Braverman and Garg, 2020]{braverman2020role}
Braverman, M. and Garg, S. (2020).
\newblock The role of randomness and noise in strategic classification.
\newblock In Roth, A., editor, {\em 1st Symposium on Foundations of Responsible
  Computing, {FORC} 2020, June 1-3, 2020, Harvard University, Cambridge, MA,
  {USA} (virtual conference)}, volume 156 of {\em LIPIcs}, pages 9:1--9:20.
  Schloss Dagstuhl - Leibniz-Zentrum f{\"{u}}r Informatik.

\bibitem[Br\"{u}ckner and Scheffer, 2011]{adversarial_games_pred}
Br\"{u}ckner, M. and Scheffer, T. (2011).
\newblock Stackelberg games for adversarial prediction problems.
\newblock In {\em Proceedings of the 17th ACM SIGKDD International Conference
  on Knowledge Discovery and Data Mining}, KDD ’11, page 547–555, New York,
  NY, USA. Association for Computing Machinery.

\bibitem[Corbett{-}Davies and Goel, 2018]{corbett2018measure}
Corbett{-}Davies, S. and Goel, S. (2018).
\newblock The measure and mismeasure of fairness: {A} critical review of fair
  machine learning.
\newblock {\em CoRR}, abs/1808.00023.

\bibitem[Dong et~al., 2018]{revealed_preferences}
Dong, J., Roth, A., Schutzman, Z., Waggoner, B., and Wu, Z.~S. (2018).
\newblock Strategic classification from revealed preferences.
\newblock In {\em Proceedings of the 2018 ACM Conference on Economics and
  Computation}, EC ’18, page 55–70, New York, NY, USA. Association for
  Computing Machinery.

\bibitem[Feng et~al., 2023]{FHL23}
Feng, Y., Hartline, J.~D., and Li, Y. (2023).
\newblock Simple mechanisms for non-linear agents.
\newblock In Bansal, N. and Nagarajan, V., editors, {\em Proceedings of the
  2023 {ACM-SIAM} Symposium on Discrete Algorithms, {SODA} 2023, Florence,
  Italy, January 22-25, 2023}, pages 3802--3816. {SIAM}.

\bibitem[Frankel and Kartik, 2019]{Frankel2019ImprovingIF}
Frankel, A.~M. and Kartik, N. (2019).
\newblock Improving information from manipulable data.
\newblock {\em arXiv: Theoretical Economics}.

\bibitem[Gaitonde et~al., 2023]{GLLLS23}
Gaitonde, J., Li, Y., Light, B., Lucier, B., and Slivkins, A. (2023).
\newblock Budget pacing in repeated auctions: Regret and efficiency without
  convergence.
\newblock In Kalai, Y.~T., editor, {\em 14th Innovations in Theoretical
  Computer Science Conference, {ITCS} 2023, January 10-13, 2023, MIT,
  Cambridge, Massachusetts, {USA}}, volume 251 of {\em LIPIcs}, pages
  52:1--52:1. Schloss Dagstuhl - Leibniz-Zentrum f{\"{u}}r Informatik.

\bibitem[Haghtalab et~al., 2020]{Haghtalab2020MaximizingWW}
Haghtalab, N., Immorlica, N., Lucier, B., and Wang, J.~Z. (2020).
\newblock Maximizing welfare with incentive-aware evaluation mechanisms.
\newblock In Bessiere, C., editor, {\em Proceedings of the Twenty-Ninth
  International Joint Conference on Artificial Intelligence, {IJCAI-20}}, pages
  160--166. International Joint Conferences on Artificial Intelligence
  Organization.
\newblock Main track.

\bibitem[Hardt et~al., 2016]{Hardt2016}
Hardt, M., Megiddo, N., Papadimitriou, C., and Wootters, M. (2016).
\newblock Strategic classification.
\newblock In {\em Proceedings of the 2016 ACM Conference on Innovations in
  Theoretical Computer Science}, ITCS ’16, page 111–122, New York, NY, USA.
  Association for Computing Machinery.

\bibitem[Harris et~al., 2021]{harris2021stateful}
Harris, K., Heidari, H., and Wu, Z.~S. (2021).
\newblock Stateful strategic regression.
\newblock In Ranzato, M., Beygelzimer, A., Dauphin, Y.~N., Liang, P., and
  Vaughan, J.~W., editors, {\em Advances in Neural Information Processing
  Systems 34: Annual Conference on Neural Information Processing Systems 2021,
  NeurIPS 2021, December 6-14, 2021, virtual}, pages 28728--28741.

\bibitem[Hu et~al., 2019]{hu2019disparate}
Hu, L., Immorlica, N., and Vaughan, J.~W. (2019).
\newblock The disparate effects of strategic manipulation.
\newblock In danah boyd and Morgenstern, J.~H., editors, {\em Proceedings of
  the Conference on Fairness, Accountability, and Transparency, FAT* 2019,
  Atlanta, GA, USA, January 29-31, 2019}, pages 259--268. {ACM}.

\bibitem[Jung et~al., 2020]{JKLPRV20}
Jung, C., Kannan, S., Lee, C., Pai, M.~M., Roth, A., and Vohra, R. (2020).
\newblock Fair prediction with endogenous behavior.
\newblock In Bir{\'{o}}, P., Hartline, J.~D., Ostrovsky, M., and Procaccia,
  A.~D., editors, {\em {EC} '20: The 21st {ACM} Conference on Economics and
  Computation, Virtual Event, Hungary, July 13-17, 2020}, pages 677--678.
  {ACM}.

\bibitem[Kleinberg et~al., 2018]{KLMR18}
Kleinberg, J., Ludwig, J., Mullainathan, S., and Rambachan, A. (2018).
\newblock Algorithmic fairness.
\newblock {\em AEA Papers and Proceedings}, 108:22--27.

\bibitem[Kleinberg and Raghavan, 2019]{Kleinberg2018HowDC}
Kleinberg, J. and Raghavan, M. (2019).
\newblock How do classifiers induce agents to invest effort strategically?
\newblock In {\em Proceedings of the 2019 ACM Conference on Economics and
  Computation}, EC '19, page 825–844, New York, NY, USA. Association for
  Computing Machinery.

\bibitem[Kleinberg, 2018]{Kleinberg18}
Kleinberg, J.~M. (2018).
\newblock Inherent trade-offs in algorithmic fairness.
\newblock In Psounis, K., Akella, A., and Wierman, A., editors, {\em Abstracts
  of the 2018 {ACM} International Conference on Measurement and Modeling of
  Computer Systems, {SIGMETRICS} 2018, Irvine, CA, USA, June 18-22, 2018},
  page~40. {ACM}.

\bibitem[Laffont and Robert, 1996]{laffont1996optimal}
Laffont, J.-J. and Robert, J. (1996).
\newblock Optimal auction with financially constrained buyers.
\newblock {\em Economics Letters}, 52(2):181--186.

\bibitem[Maskin, 2000]{maskin2000auctions}
Maskin, E.~S. (2000).
\newblock Auctions, development, and privatization: Efficient auctions with
  liquidity-constrained buyers.
\newblock {\em European Economic Review}, 44(4):667--681.

\bibitem[Miller et~al., 2020]{Miller2019StrategicCI}
Miller, J., Milli, S., and Hardt, M. (2020).
\newblock Strategic classification is causal modeling in disguise.
\newblock In {\em Proceedings of the 37th International Conference on Machine
  Learning, ICML 2020, 13-18 July 2020, Virtual Event}, volume 119 of {\em
  Proceedings of Machine Learning Research}, pages 6917--6926. PMLR.

\bibitem[Milli et~al., 2019]{Milli2018TheSC}
Milli, S., Miller, J., Dragan, A.~D., and Hardt, M. (2019).
\newblock The social cost of strategic classification.
\newblock In {\em Proceedings of the Conference on Fairness, Accountability,
  and Transparency}, FAT* '19, page 230–239, New York, NY, USA. Association
  for Computing Machinery.

\bibitem[Pai and Vohra, 2014]{pai2014optimal}
Pai, M.~M. and Vohra, R. (2014).
\newblock Optimal auctions with financially constrained buyers.
\newblock {\em J. Econ. Theory}, 150:383--425.

\bibitem[Rambachan et~al., 2020]{RKLM20}
Rambachan, A., Kleinberg, J., Ludwig, J., and Mullainathan, S. (2020).
\newblock An economic perspective on algorithmic fairness.
\newblock {\em AEA Papers and Proceedings}, 110:91--95.

\bibitem[Stiglitz and Weiss, 1981]{SW-81}
Stiglitz, J. and Weiss, A. (1981).
\newblock Credit rationing in markets with imperfect information.
\newblock {\em American Economic Review}, 71(3):393--410.

\bibitem[Stiglitz, 1975]{sti-75}
Stiglitz, J.~E. (1975).
\newblock {The Theory of \&quot;Screening,\&quot; Education, and the
  Distribution of Income}.
\newblock {\em American Economic Review}, 65(3):283--300.

\bibitem[Xiao et~al., 2020]{xiao2020optimal}
Xiao, S., Wang, Z., Chen, M., Tang, P., and Yang, X. (2020).
\newblock Optimal common contract with heterogeneous agents.
\newblock {\em Proceedings of the AAAI Conference on Artificial Intelligence},
  34(05):7309--7316.

\end{thebibliography}

 \newpage
\begin{appendix}
\section{Deferred Proofs of Propositions and Lemmas}
\label{a:appendix_a}

\subsection{Proof that a Threshold Mechanism is Optimal for Single-budget}
\label{a:proof_thresh}

\deterministic* 
\begin{proof}
The optimality of $\balloc^{\quath}$ in fact follows from the stronger statement in \Cref{lem:threshoptsinglebudgoffline} (below).
\end{proof}

\noindent The offline optimal mechanism (\Cref{def:offopt}) is generally unachievable.  Despite that caveat, it is possible to achieve the offline optimal mechanism for the special case of singular budgets, as subsequently stated in \Cref{lem:threshoptsinglebudgoffline}.  Recall the intuition given in the main body of the paper: ``\Cref{prop:threshoptsinglebudg} holds because single-budget is a simple setting in which quality-thresholds directly implement skill-thresholds, in particular for the principal's threshold $\thresh$."

\begin{lemma}
\label{lem:threshoptsinglebudgoffline}
Assume that an agent has constant budget $\bar{\budg}$ on effort, i.e., the distribution $\bdist$ is a singleton point mass.  {\em Without directly observing the realization of the agent's skill $\hat{\skill}\sim\sdist$}, the threshold mechanism $\balloc^{\quath}$ with $\quath = \thresh\cdot\bar{\budg}$ is {\em offline optimal}.
\end{lemma}
\begin{proof}

We will show that $\balloc^{\quath}$ is offline optimal by showing that it gives allocation 1 to every (randomized) agent skill-type which gives positive utility to the principal, and gives allocation 0 to every agent skill-type which gives negative utility to the principal, thus pointwise-maximizing the principal's utility function.

For agent $\agent=(\skill, \bar{\budg})$, the minimum effort required to reach threshold $\quath$ is $\effth = \sfrac{\quath}{\skill}= \left(\sfrac{\thresh}{\skill}\right)\bar{\budg}$.  Then $\eff'\leq\bar{\budg}$ is affordable (and rational) for $\agent$ if and only if $\sfrac{\thresh}{\skill}\leq 1$, an inequality which itself is true if and only the principal's utility $\skill-\thresh\geq 0$ (from admitting $\agent$; see page~\pageref{eqn:princproblem}).  By setting the quality threshold to be the maximum achievable by the skill level $\thresh$ (which corresponds to 0-utility for skill-type $\thresh$), the mechanism allocates to exactly the upward closed set of all agent types from which it receives positive utility (\Cref{fact:highskillfirst}), and no others.\qedhere

\end{proof}

\subsection{Example of Insufficiency of Deterministic Mechanisms}
\label{a:exthreshinsufficient}

The following \Cref{ex:threshinsuff} provides the proof-by-counterexample for \Cref{prop:threshinsuff}.

\begin{example}
\label{ex:threshinsuff}
Admission game admission game $\game = (\sdist, \bdist, \thresh)$ is defined as follows.  

Agent $\agent$ has discrete skill space with two types (i.e, $|\sdist|=2$) with low skill $\skilli[\lo] = 1+\eps_{\lo}$ (for $\eps_{\lo}\rightarrow0$) and high skill $\skilli[\hi] = 2$.  Agent $\agent$ has discrete budget space with two types ($|\bdist|=2$) with low budget $\budg_{\lo} = \sfrac{1}{2}$ and high budget $b_{\hi} = 1$.  The distributions $\sdist$ and $\bdist$ have positive mass on each element of their respective supports but otherwise we leave them indeterminate.  The principal $\princ$'s skill threshold to measure utility is $\thresh=\sfrac{3}{2}$.
\end{example}

\noindent The following analysis will show that for the setting of \Cref{ex:threshinsuff}, all deterministic threshold mechanisms are dominated by stochastic allocation $\alloc=(1-\sfrac{\eps_{\lo}}{1+\eps_{\lo}}-\eps_{\alloc})$
for agents exhibiting quality at least 1.

As a starting point, consider the deterministic threshold mechanism $\balloc^{\quath}$ which picks $\quath = 1$.  The following gives initial analysis of an agent with high skill type $\skill_{\hi}$:

\begin{itemize}
    \item minimum effort to achieve $\quath$ is: $\eff_{\hi}=\sfrac{1}{2}$;
    \item utility from achieving threshold $\quath$ is: $1-\sfrac{1}{2}=\sfrac{1}{2}$;
    \item an agent of type $(\skill_{\hi},\budg_{\hi})$ will put in effort to be admitted (given $\quath =1$ and furthermore, whenever $\quath < 2$);
    \item an agent of type $(\skill_{\hi},\budg_{\lo})$ will also put in effort to be admitted, but critically, can not put in effort to be admitted if $\quath$ is increased above 1 by any $\eps_{\qua}>0$ because this agent-type $(\skill_{\hi},\budg_{\lo})$ is bounded by maximum quality $\quamax=\skill_{\hi}\cdot\budg_{\lo} = 2\cdot\sfrac{1}{2}=1$.
\end{itemize}
    
\noindent Alternatively, the following gives initial analysis of an agent with low skill type $\skill_{\lo}$:
\begin{itemize}
    \item minimum effort to achieve $\quath$ is: $\eff_{\lo}=\sfrac{1}{1+\eps_{\lo}}$;
    \item utility from achieving threshold $\quath$ is: $1-\sfrac{1}{1+\eps_{\lo}} = \sfrac{\eps_{\lo}}{1+\eps_{\lo}}$;
    \item an agent of type $(\skill_{\lo},\budg_{\hi})$ will put in effort to be admitted (given $\quath =1$);
    \item an agent of type $(\skill_{\lo},\budg_{\lo})$ will put in 0 effort (because maximum quality is less than $\quath$).
\end{itemize}

\noindent Offline-optimal (\Cref{def:offopt}) allocates all agents with skill $(\skill_{\hi},\cdot)$ and rejects all agents $(\skill_{\lo},\cdot)$.  The current quality threshold under consideration $\quath=1$ is the largest threshold that will admit types $(\skill_{\hi},\budg_{\lo})$.  Let $\dprob_{\att,\tier}$ be the probability corresponding to arbitrary agent type-attribute $\att\in\{\skill,\budg\}$ and tier $\tier\in\{\lo,\hi\}$.  The performance of every threshold mechanism fails to approach the performance of offline optimal (we write `$\gg$' to indicate that the gap is bounded away from 0):
\begin{itemize}
    \item thresholds $\quath_+>\quath=1$ will not admit types $(\skill_{\hi},\budg_{\lo})$ and thus will additively underperform offline optimal by at least:  $$\dprob_{\skill,\hi}\cdot\dprob_{\budg,\lo}\cdot\left( \skill_{\hi}-\thresh\right)= \dprob_{\skill,\hi}\cdot\dprob_{\budg,\lo}\cdot\left( \sfrac{1}{2}\right)\gg0$$
    \item thresholds $\quath_-\leq \quath =1$ will admit types $(\skill_{\lo},\budg_{\hi})$ and thus will additively underperform offline optimal by at least:
    $$\dprob_{\skill,\lo}\cdot\dprob_{\budg,\hi}\cdot\left( \thresh-\skill_{\lo}\right)= \dprob_{\skill,\lo}\cdot\dprob_{\budg,\hi}\cdot\left( \sfrac{1}{2}-\eps_{\lo}\right)\gg0$$
\end{itemize}

\noindent However, if we maintain $\quath=1$ and rather {\em decrease the probability of allocation} from 1 to $(1-\sfrac{\eps_{\lo}}{1+\eps_{\lo}}-\eps_{\alloc})$ for $\eps_{\alloc}\rightarrow0$, then all high types still strictly put in effort and will be admitted (with near-certainty), but the low types now strictly prefer to put in 0 effort.

Formally, for (single-menu-option) allocation $\alloc=\balloc(1)=(1-\sfrac{\eps_{\lo}}{1+\eps_{\lo}}-\eps_{\alloc})$, the utility calculations are (assuming minimum effort to be admitted, ignoring affordability due to budget):
\begin{itemize}
    \item agents with high skill type $\skill_{\hi}$ have utility: $\util_{(\hi,\cdot)}=\balloc(1)-\eff_{\hi}=\sfrac{1}{2}-\sfrac{\eps_{\lo}}{1+\eps_{\lo}}-\eps_{\alloc}>0$;
    \item agents with low skill type $\skill_{\lo}$ have utility: $\util_{(\lo,\cdot)}= \balloc(1)-\eff_{\lo}  =-\eps_{\alloc}<0$.
\end{itemize}

\noindent Considering, $\eps_{\lo}\rightarrow0$ and $\eps_{\alloc}\rightarrow0$, the admission-rate of high-skill agents approaches 1 and thus the expected performance of this mechanism becomes arbitrarily close to the performance of offline optimal.  Therefore, it strictly improves on the best of any deterministic threshold mechanism (which can't approach performance of offline optimal by the analysis above).

This completes the counterexample to illustrate that deterministic mechanisms are not sufficient.

\subsection{A Comparison of Slanted-Stair Allocation to Deterministic Threshold}
\label{a:compare-ss-to-dt}

\Cref{s:discussmainthm} gives discussion of the optimal characterization of mechanisms in \Cref{thm:2skill_nbudget}.  For purposes of space, we complete here the discussion of alignment between optimal mechanisms and {\em fairness}.  To summarize the initial discussion in the main body, this alignment first is observed intuitively from the structure of the principal's utility from admitting an agent $\agent$, which is $\util_{\princ}(\agent~|~\text{admitted})=\skill-\thresh$.  Second, the offline optimal allocation is the ``perfect" mechanism performance and is also arguably an ideal allocation in terms of fairness.  We now give an additional intuitive perspective supporting this alignment.

Third -- analyzing qualitatively for both mechanism performance and fairness -- we can make a comparison between (a) an optimal single-jump-at-$\qua^{\text{jump}}$, slanted-stair allocation rule $\balloc^*$ of \Cref{thm:2skill_nbudget}; and (b) the specific -- albeit modified -- threshold mechanism that jumps from allocation 0 to 1 at the same quality $\qua^{\text{jump}}$.  For convenience, we copy \Cref{fig:opt_2xn}(Top) 
into \Cref{fig:opt_2xn-fair-analysis}.

The modification is that the threshold mechanism in this section will require for admission that an agent's exhibited quality be {\em strictly greater} than the threshold.  This organizes the closed-versus-open endpoints of the threshold-step in a way that allows for a more-direct comparison to slanted-stair functions.  This is illustrated in \Cref{fig:opt_2xn-fair-analysis}(Bottom).

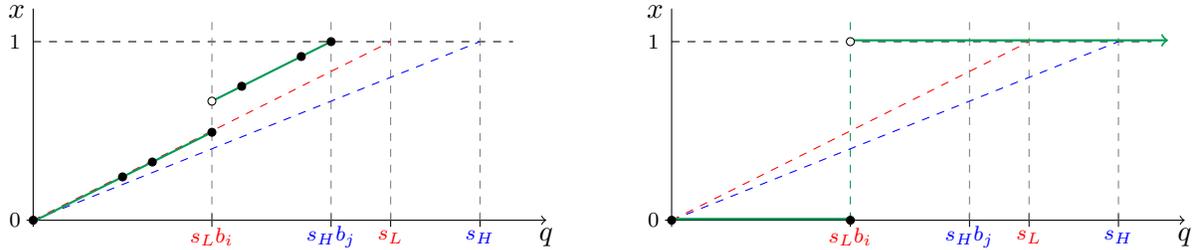
\begin{figure}[t]
\begin{center}
\begin{tikzpicture}[scale = 0.44,pile/.style={->}]

\draw [pile] (-0.2,0) -- (15.5, 0);
\draw (0, -0.2) -- (0, 6.4);

\draw (13.50,-0.2) -- (13.50,0.2);
\draw (5.4,-0.2) -- (5.4,0.2);
\draw (9.0,-0.2) -- (9.0,0.2);
\draw (10.8,-0.2) -- (10.8,0.2);

\draw [dashed] (0.0,5.4) -- (14.5,5.4);

\draw [gray, dashed] (5.4,0) -- (5.4,6.0);
\draw [gray, dashed] (9.0,0) -- (9.0,6.0);
\draw [gray, dashed] (10.8,0) -- (10.8,6.0);
\draw [gray, dashed] (13.5,0) -- (13.5,6.0);

\draw[blue, dashed] (0,0) -- (13.5,5.4);
\draw[red, dashed] (0.0,0.0) -- (10.8,5.4);
\draw[ForestGreen, thick] (0,-0.04) -- (5.4,2.66);
\draw[ForestGreen, thick] (5.4,3.6) -- (9.0, 5.4);

\draw[fill=black] (0,0) circle (0.12cm);
\draw[fill=black] (2.7,1.31) circle (0.12cm);
\draw[fill=black] (3.6,1.76) circle (0.12cm);
\draw[fill=black] (5.4,2.66) circle (0.12cm);

\draw[fill=white] (5.4, 3.6) circle (0.12cm);
\draw[fill=black] (6.3,4.05) circle (0.12cm);
\draw[fill=black] (8.1,4.95) circle (0.12cm);
\draw[fill=black] (9.0,5.4) circle (0.12cm);

\draw (13.5,-0.5) node {\figsz \textcolor{blue}{$\skilli[\hi] $}};
\draw (15.5, -0.5) node {$q$};
\draw (-0.5, 6.3) node {$x$};

\draw (-0.55, 0) node {\figsz $0$};
\draw (-0.55, 5.4) node {\figsz $1$};

\draw (5.4,-0.5) node {\figsz \textcolor{red}{$\skilli[\lo]\budgi$}};
\draw (10.8, -0.5) node {\figsz \textcolor{red}{$\skilli[\lo]$}};
\draw (9.0,-0.5) node {\figsz \textcolor{blue}{$\skilli[\hi]\budgi[j]$}};

\end{tikzpicture}
\hspace{0.6cm}
\begin{tikzpicture}[scale = 0.44,pile/.style={->}]

\draw [pile] (-0.2,0) -- (15.5, 0);
\draw (0, -0.2) -- (0, 6.4);

\draw (13.50,-0.2) -- (13.50,0.2);
\draw (5.4,-0.2) -- (5.4,0.2);
\draw (9.0,-0.2) -- (9.0,0.2);
\draw (10.8,-0.2) -- (10.8,0.2);

\draw [dashed] (0.0,5.4) -- (14.5,5.4);

\draw [ForestGreen, dashed] (5.4,0) -- (5.4,6.0);
\draw [gray, dashed] (9.0,0) -- (9.0,6.0);
\draw [gray, dashed] (10.8,0) -- (10.8,6.0);
\draw [gray, dashed] (13.5,0) -- (13.5,6.0);

\draw[blue, dashed] (0,0) -- (13.5,5.4);
\draw[red, dashed] (0.0,0.0) -- (10.8,5.4);
\draw[ForestGreen, thick] (0,0.04) -- (5.4,0.04);
\draw[pile, ForestGreen, thick] (5.4,5.44) -- (15.0, 5.44);

\draw[fill=black] (0,0) circle (0.12cm);

\draw[fill=black] (5.4, 0) circle (0.12cm);
\draw[fill=white] (5.4, 5.4) circle (0.12cm);

\draw (13.5,-0.5) node {\figsz \textcolor{blue}{$\skilli[\hi] $}};
\draw (15.5, -0.5) node {$q$};
\draw (-0.5, 6.3) node {$x$};

\draw (-0.55, 0) node {\figsz $0$};
\draw (-0.55, 5.4) node {\figsz $1$};

\draw (5.4,-0.5) node {\figsz \textcolor{red}{$\skilli[\lo]\budgi$}};
\draw (10.8, -0.5) node {\figsz \textcolor{red}{$\skilli[\lo]$}};
\draw (9.0,-0.5) node {\figsz \textcolor{blue}{$\skilli[\hi]\budgi[j]$}};

\end{tikzpicture}
\end{center} 
\caption{(Top) Illustration of a one-jump, slanted-stair allocation curve $\balloc^*$ (solid green), which is assumed to be optimal for its game parameters (for analysis purposes). The black dots are an example of discrete menu options.  The single jump occurs at $\qua^{\text{jump}} = \skilli[\lo]\budgi$.  Regarding discussion in \Cref{a:compare-ss-to-dt}: the first, left-most region is ``below the jump;" the second, middle region is ``above the jump but not fully allocated;" and the third, right-most region is ``full allocation." (Bottom) The (strictly-greater-than) threshold mechanism with threshold set equal to the jump-point in (Top), i.e., with $\quath=\qua^{\text{jump}}=\skilli[\lo]\budgi$.}
\label{fig:opt_2xn-fair-analysis}
\end{figure}

Graphically, the optimal mechanism $\balloc^*$ (which may be a discrete menu, corresponding to our discrete setting) will qualitatively have the single-jump structure of \Cref{fig:opt_2xn-fair-analysis}(Top)
.  Using agent skill/budget-indexing of \Cref{fig:opt_2xn-fair-analysis} (i.e., notation), the  general structure of $\balloc^*$ has three regions:
\begin{enumerate}
\item the left-most region is ``below the jump" defined by qualities $\qua\in[0,\qua^{\text{jump}}=\skilli[\lo]\budgi[j]]$;
\item the middle region is ``above the jump but not full allocation" defined by qualities $\qua\in(\qua^{\text{jump}}=\skilli[\lo]\budgi[j],\skilli[\hi]\budgi[\lo])$;
\item the right-most region is ``full allocation" defined by the quality $\qua=\skilli[\hi]\budgi[\lo]$ (and all larger qualities, though rational agents never choose these larger levels, which require exerting superfluous effort to achieve, without an increase in allocation).
\end{enumerate}
\noindent The allocation rule $\balloc^*$ is optimal for the standard principal-objective, so it obviously dominates the threshold mechanism with its quality-space threshold set to be $\quath=\qua^{\text{jump}} = \skilli[\lo]\budgi[j]$.  In the following discussion, agents are considered to be ``in" the region which contains their optimally-chosen quality for the given mechanism (subject to tie-breaking).  We qualitatively analyze the same comparison for fairness:
\begin{enumerate}
    \item in the left-most region, low-skill agents receive 0-allocation according to both $\balloc^*$ and the threshold mechanism; by contrast, high-skill agents receive 0-allocation according to the threshold mechanism, but positive allocation according to $\balloc^*$ (the solid green line in \Cref{fig:opt_2xn-fair-analysis}(Top)
    ; we suggest in this first region -- regardless of the choice of fairness measure -- that the fairness of $\balloc^*$ dominates the fairness of the threshold mechanism;
    \item in the middle region, {\em all} low-skill agents receive allocation $\balloc^*(\qua)$ for $\qua\rightarrow(\qua^{\text{jump}})^+$ (from above) according to $\balloc^*$ (by tie-breaking), which for {\em all} low-skill agents increases to full-allocation of 1 according to the threshold mechanism; whereas each high-skill agent $(\skilli[\hi],\budgi[j]$ is exhibiting its respective maximum quality $\quamax_{\hi,j}=\skilli[\hi]\cdot\budgi[j]$ and receives allocation $\balloc^*(\quamax_{\hi,j})$ according to $\balloc^*$ which increases to full-allocation of 1 according to the threshold mechanism;\\
    in this second region, we can not make a dominance argument because it partially depends on the unknown densities of agent-types represented in this region and it also depends on the technical measure of fairness; however, ignoring expectation and proportional density and instead simply comparing agents one-to-one, we do observe that low-skill types receive the larger benefit (increase in allocation) if we start with $\balloc^*$ as our default mechanism and consider changing to the threshold mechanism; furthermore, the threshold mechanism abolishes the (properly oriented) cardinal difference between low-skill and high-skill agents by instead awarding them an ``arguably unfair" constant allocation (of 1);
    \item in the right-most region, all skill-types in all mechanisms receive the same allocation of 1; thus in this third region, the mechanism $\balloc^*$ and the threshold mechanism are equally fair (or equally unfair).
\end{enumerate}
\noindent Intuitively, the preceding comparison between the optimal mechanism $\balloc^*$ and the threshold mechanism -- which specifically have jumps at the same quality-index $\qua^{\text{jump}}$ -- suggests that (single-jump) slanted-stair mechanisms are indeed more fair.  In fact, we have already stated a strict dominance relationship for an obvious, special-case choice of the technical fairness measure.

\princasfairmeasure*

\noindent Recall, the principal is naturally aligned with fairness. Then if we assign the fairness measure to be equal to the utility function of the principal, the analysis of the optimal mechanism for fairness gives the identical result as \Cref{thm:2skill_nbudget}.


\end{appendix}

\end{document}